\documentclass[11pt,a4paper,reqno]{amsart}  
\usepackage{mathtools, amsmath, amsthm, amssymb, xcolor, hyperref, graphicx}
\newtheorem{lemma}{Lemma}

\makeatletter
\def\hlinew#1{
	\noalign{\ifnum0=`}\fi\hrule \@height #1 \futurelet
	\reserved@a\@xhline}
\makeatother

\renewcommand{\arraystretch}{1.27}

\allowdisplaybreaks[1]

\theoremstyle{plain}
\newtheorem{thm}{Theorem} 

\newtheorem{cor}{Corollary}
\newtheorem{prop}{Proposition}
\newtheorem{conj}{Conjecture}
\theoremstyle{definition}
\newtheorem{remark}{Remark}

\newtheorem*{thm*}{Theorem}

\usepackage{graphicx,color}

\newcommand{\QQ}{\mathbb{Q}}
\newcommand{\ZZ}{\mathbb{Z}}

\newcommand{\p}{\partial}

\newcommand{\X}{\mathcal{X}}
\newcommand{\bl}{\lambda}
\newcommand{\bmu}{\boldsymbol{\mu}}

\def\beq{\begin{equation}}
\def\eeq{\end{equation}}
\def\be{\begin{equation}}
\def\ee{\end{equation}}
\def\bes{\begin{equation*}}
\def\ees{\end{equation*}}
\def\thin{\hspace 0.5pt}
\newcommand{\nn}{\nonumber}
\def\={\;=\;}   \def\+{\,+\,} \def\m{\,-\,}  \def\:{\;:=\;}
\def\l{\lambda}    \def\a{\alpha} \def\dd{\mathbf d}   \def\g{\gamma} 
 \def\Q{\mathbb Q}  \def\tr{{\rm tr}}  
\def\B#1#2{\langle#1\rangle^\Theta_#2}    % macro for BGW numbers
\def\thin{\hskip 1 pt}  
           
\begin{document}
 \title[Combinatorics and asymptotics of BGW numbers]{Combinatorics and large genus asymptotics of the Br\'ezin--Gross--Witten numbers}
\author{Jindong Guo, Paul Norbury, Di Yang, Don Zagier}

\address{Jindong Guo, School of Mathematical Sciences, University of Science and Technology of China, 230026 Hefei, P.R.~China}
\email{guojindong@mail.ustc.edu.cn}
\address{Paul Norbury, School of Mathematics and Statistics, University of Melbourne, VIC 3010 Melbourne, Australia}
\email{norbury@unimelb.edu.au}
\address{Di Yang, School of Mathematical Sciences, University of Science and Technology of China, 230026 Hefei, P.R.~China}
\email{diyang@ustc.edu.cn}
\address{Don Zagier, Max Planck Institute for Mathematics, 53111 Bonn, 
Germany, and International Centre for Theoretical Physics, Trieste, Italy}
\email{dbz@mpim-bonn.mpg.de}
\begin{abstract}
In this paper, we study combinatorial and asymptotic properties of 
some interesting rational numbers called the Br\'ezin--Gross--Witten (BGW) 
numbers, which can be represented as the intersection numbers of psi and 
Theta classes on the moduli space of stable algebraic curves. In particular, 
we discover and prove the uniform large genus leading asymptotics of certain 
normalized BGW numbers, and give a new proof of the polynomiality 
phenomenon for the large genus asymptotics. We also propose, with extensive numerical data, several new 
conjectures including monotonicity and integrality on the BGW numbers. 
Applications to the Painlev\'e II hierarchy and to the BGW-kappa numbers 
are given.
\end{abstract}
\maketitle
\tableofcontents

\section{Introduction}\label{secintro}
In this paper, we study some interesting and important rational numbers  
called the {\it Br\'ezin--Gross--Witten (BGW) numbers} \cite{A,BR,BG,DN,DYZ,GW}. 
Originally, the BGW numbers were defined via matrix models \cite{BG,GW},
and specifically are proportional to the Taylor coefficients with respect to the so-called Miwa variables (cf.~\cite{AC,BG, GN, GW, Kontsevich}) of the logarithm of the integral 
\beq
\int_{U_n} e^{\frac1 \beta {\rm tr} (J^\dagger U+J U^\dagger)}\, dU\,,
\eeq
where $dU$ denotes the normalized Haar measure on the unitary group~$U_n$, 
and $J$ and $J^\dagger$ are arbitrary $n\times n$ matrices.  
Later, alternative definitions and properties 
of the BGW numbers were 
given in a number of further papers (cf.~\cite{A, BR,CGG, DN,DYZ, KN, MMS, Norbury0,YZ}).
 As customary in the literature, denote by $\langle \tau_{d_1} \dots \tau_{d_n} \rangle^{\Theta}_g$ the BGW numbers, 
 where $g\ge1$ (genus), $n\ge1$ and $d_1,\dots,d_n\ge0$ satisfy
 \begin{align}\label{ddc}
d_1\+\dots\+d_n\=g-1
 \end{align}
(see e.g.~\cite{A,DN,YZ}).

For a long time, no topological or combinatorial meaning for the BGW  numbers was known, but recently, two ways 
were found to define these numbers topologically. 
First of all, they are equal to the following integrals on the moduli space of stable curves:
\beq
\langle\tau_{d_1}\dots\tau_{d_n}\rangle^{\Theta}_g
\= \int_{\overline{\mathcal{M}}_{g,n}}\,\psi_1^{d_1}\cdots\psi_n^{d_n} \, \Theta_{g,n}\,,
\eeq
as it was conjectured in~\cite{Norbury0} with later a complete proof given in~\cite{CGG}.
Here $\overline{\mathcal{M}}_{g,n}$ denotes the Deligne--Mumford moduli space of stable algebraic curves of genus $g$
with $n$ distinct marked points, $\psi_j$ denotes the first Chern class of the $j$th cotangent line bundle, and $\Theta_{g,n}$ 
denotes the Theta-class introduced by the second author of the present paper~\cite{Norbury0} (cf.~\cite{CGG, KN, XY}). 
This definition, which is the reason for the notation we use, is in exact analogy with Witten's notation for his
intersection numbers~\cite{Witten} and elucidate~\eqref{ddc} simply as the degree-dimension counting.
Secondly, it was proved in~\cite{YZ} that the BGW numbers can be given 
by an ELSV-like formula
\beq
\langle\tau_{d_1}\dots\tau_{d_n}\rangle_g^{\Theta} 
\= \frac{(-1)^{g-1+n} \, 2^{2g-2+n}}{\prod_{i=1}^n d_i!} \,\int_{\overline{\mathcal{M}}_{g,n}} 
\frac{\Lambda(-1)^2\,\Lambda(\tfrac12)\,\exp\bigl(\thin\sum_{d=1}^\infty \!\frac{(-1)^{d-1}\thin\kappa_d}{2^dd}\bigr)}
{\prod_{i=1}^n\bigl(1+\frac{2d_i+1}{2}\psi_i\bigr)} \,,
\eeq
where $\Lambda(z)$ denotes the Chern polynomial of the Hodge bundle and $\kappa_d:=f_*(\psi_{n+1}^{d+1})$
are the kappa classes, with $f:\overline{\mathcal{M}}_{g,n+1}\to \overline{\mathcal{M}}_{g,n}$ being the forgetful map.
Two efficient algorithms of computing the BGW numbers 
will be reviewed in Section~\ref{secgeneral}.

The first few BGW numbers are given by
\beq
\begin{split}
&\B{\tau_0}1=\frac18\,, \;\quad \B{\tau_1}2=\frac{3}{128}\,, \;\quad
 \B{\tau_1^2}3=\frac{63}{512}\,,\;\quad \B{\tau_2}3=\frac{15}{1024}\,, \\
&\B{\tau_1^3}4=\frac{7221}{2048}\,,\;\quad \B{\tau_1\tau_2}4=\frac{8625}{32768}\,,\;\quad \B{\tau_3}4=\frac{525}{32768}\,.
\end{split}
\label{BGWnumbers}
\eeq
(Here we have omitted the numbers containing $\tau_0$ except for $\langle\tau_0\rangle^\Theta_1$
because of equation~\eqref{CC0} below.)  
From these and many further examples, we observe that
 the BGW numbers $\langle\tau_{d_1}\cdots\tau_{d_n}\rangle^{\Theta}_g$, $d_1,\dots,d_n\ge0$, are integral away from the prime~2,
and we conjecture that this is true in general. We call this the Integrality Conjecture.
Furthermore, they seem to have many small factors, e.g.,
$ \langle\tau_2\tau_3\rangle_6^{\Theta}$ equals $2^{-21}\thin 3^2\thin 5^2\thin 7^3\thin 103^1$, 
and  $\langle\tau_2^3\tau_3^2\tau_4\tau_5\rangle_{22}^{\Theta}$ is divisible by $2^{-71}\thin 3^{11}\thin 5^2\thin 7^2\thin 11^2$.
A precise conjecture that at least partially explains these factorizations will be given 
in Section~\ref{secallbutonepartfixed} (see Conjecture~\ref{conjfactorsBGW1}).
	
To proceed, let us introduce the {\it normalized BGW numbers} $C(\dd)$ by
\beq\label{defC}
C(\mathbf{d}) \:  \frac{2^{2g(\dd)-1}\thin \prod_{j=1}^n (2d_j+1)!!}{(X(\mathbf{d})-1)!}\, \langle\tau_{d_1}\cdots\tau_{d_n}\rangle_{g(\dd)}^{\Theta}\,,
	\eeq
where $\mathbf{d}=(d_1,\dots,d_n)\in(\mathbb{Z}_{\ge0})^n$, $g(\dd)=|\mathbf{d}|+1$ with $|\mathbf{d}|:=d_1+\cdots+d_n$, and
\begin{align}\label{defXd}
X(\mathbf{d}) \: \thin\sum_{j=1}^n(2d_j+1) \= 2 \,g(\dd)-2+n\,.
\end{align}
Obviously, $C(0)=1/4$.
Using a relation given in Section~\ref{secDN} (see~\eqref{BB0}), we know that 
\beq\label{CC0}
C(\dd)\=C(0,\dd)\,.
\eeq
This property would of course also be true without the factor~$2^{2g(\dd)-1}$ in~\eqref{defC}, but the
normalization given here will make the asymptotic properties of the numbers nicer.  

Because of~\eqref{CC0}, in the study of $C(\mathbf{d})$, it is sufficient to consider the case 
when $d_1,\dots,d_n$ are all positive, in other words, when ${\bf d}$ is a partition. 
From now on, we will usually restrict to this case. In particular, we do this in the following table, 
which gives the values of the normalized BGW numbers for $g=2,\dots,7$.

\begin{table}[phbt]
	\begin{center}
		\begin{tabular}{|l|c|l|r|}
			\hlinew{1.5pt}
			%partition $\lambda$&$C_{\lambda}$& appr. of $C_{\lambda}$&$ C_{\lambda}$ mtp by t no.\\
			%\hlinew{1.5pt}
			\multicolumn{4}{|c|}{$g=2,\quad D=32$}\\
			\hlinew{1.5pt}
			$(1)$&$\frac{9}{32}$&$0.281250$&$9$\\
			\hlinew{1.5pt}
			\multicolumn{4}{|c|}{$g=3,\quad D=1280$}\\
			\hlinew{1.5pt}
			$(2)$&$\frac{75}{256}$&$0.292969$&$375$\\
			\hline
			$(1,1)$&$\frac{189}{640}$&$0.295313$&$378$\\
			\hlinew{1.5pt}
			\multicolumn{4}{|c|}{$g=4,\quad D=143360$}\\
			\hlinew{1.5pt}
			$(3)$&$\frac{1225}{4096}$&$0.299072$&$42875$\\
			\hline
			$(1,2)$&$\frac{8625}{28672}$&$0.300816$&$43125$\\
			\hline
			$(1,1,1)$&$\frac{21663}{71680}$&$0.302218$&$43326$\\
			\hlinew{1.5pt}
			\multicolumn{4}{|c|}{$g=5,\quad D=378470400$}\\
			\hlinew{1.5pt}
			$(4)$&$\frac{19845}{65536}$&$0.302811$&$114604875$\\
			\hline
			$(1,3)$&$\frac{14945}{49152}$&$ 0.304057$&$115076500$\\
			$(2,2)$&$\frac{209275}{688128}$&$0.304122$&$115101250$\\
			\hline
			$(1,1,2)$&$\frac{34995}{114688}$&$0.305132$&$115483500$\\
			\hline
			$(1,1,1,1)$&$\frac{4825971}{15769600}$
			&$0.306030$&$115823304$\\
			\hlinew{1.5pt}
			\multicolumn{4}{|c|}{$g=6,\quad D=91842150400$}\\
			\hlinew{1.5pt}
			$(5)$&$\frac{160083}{524288}$&$ 0.305334$&$28042539525$\\
			\hline
			$(1,4)$&$\frac{1766205}{5767168}$
			&$0.306252$&$28126814625$\\
			$(2,3)$&$\frac{883225}{2883584}$
			&$0.306294$&$28130716250$\\
			\hline
			$(1,1,3)$&$\frac{442715}{1441792}$
			&$0.307059$&$28200945500$\\
			$(1,2,2)$&$\frac{6198625}{20185088}$
			&$0.307089$&$28203743750$\\
			\hline
			$(1,1,1,2)$&$\frac{5768625}{18743296}$&$0.307770$&$28266262500$\\
			\hline
			$(1,1,1,1,1)$&$\frac{3540311739}{11480268800}$&$0.308382$&$28322493912$\\
			\hlinew{1.5pt}
			\multicolumn{4}{|c|}{$g=7,\quad D=37471597363200$}\\
			\hlinew{1.5pt}
			$(6)$&$\frac{1288287}{4194304}$&$0.307152$&$11509459436475$\\
			\hline
			$(1,5)$&$\frac{8392923}{27262976}$&$0.307851$&$11535653017350$\\
			$(2,4)$&$\frac{184659615}{599785472}$&$0.307876$&$11536609447125$\\
			$(3,3)$&$\frac{138495805}{449839104}$&$0.307879$&$11536700556500$\\
			\hline
			$(1,1,4)$&$\frac{92508885}{299892736}$&$0.308473$&$11558985180750$\\
			$(1,2,3)$&$\frac{46257505}{149946368}$&$0.308494$&$11559750499500$\\
			$(2,2,2)$&$\frac{4533499725}{14694744064}$&$0.308512$&$11560424298750$\\
			\hline
			$(1,1,1,3)$&$\frac{23168971}{74973184}$&$0.309030$&$11579851705800$\\
			$(1,1,2,2)$&$\frac{2270671055}{7347372032}$&$0.309045$&$11580422380500$\\
			\hline
			$(1,1,1,1,2)$&$\frac{1137113661}{3673686016}$&$0.309529$&$11598559342200$\\
			\hline
			$(1,1,1,1,1,1)$&$\frac{34568613873}{111522611200}$&$0.309970$&$11615054261328$\\
			\hline
		\end{tabular}\vskip 7pt
	\end{center}
	\caption{Numerical data for $C(\dd)$ with $g\le 7$} \label{tablenormalizednumbers}
\end{table}

From Table~\ref{tablenormalizednumbers} we observe that the normalized BGW numbers $C(\mathbf{d})$ for partitions of $g-1$ 
with $2\le g\le 7$ all lie between the values for the crudest and finest partitions $(g-1)$ and $(1^{g-1})$ of~$g-1$. 
(Here we use the standard convention of writing $d^m$ to mean that the argument~$d$ is repeated $m$ times.)
By a computer program using an algorithm given in Section~\ref{secDN}, 
we also checked that this is true up to $g= 40$. For example, for $g=40$, all values lie between the two numbers
\beq\label{exampleg40}
C(39)=0.316326705\cdots\,, \qquad C(1^{39})=0.316963758\cdots\,. 
\eeq
 The following conjecture states that this nesting property holds for all~$g$.

\begin{conj}\label{conjnesting}
	We have $C(g-1)\le C(\dd)\le C(1^{g-1})$ for any partition~$\dd$ of~$g-1$.
\end{conj} 

But in fact much more is true. We denote by~$\ell({\bf d})$ the length of a partition~{\bf d} 
and for any fixed $g\ge1$ we define an ordering for all partitions of~$g-1$ first by increasing
length and then lexicographically for a given length, i.e., 
$\dd\prec\dd'$ if either $\ell(\dd)<\ell(\dd')$ or $\ell(\dd)=\ell(\dd')$ and
$\dd_i<\dd_i'$, where the non-zero entries of both $\dd$ and~$\dd'$ are arranged in 
increasing order and $i$ is the first index for which $\dd_i\ne\dd_i'$. 
Purely by chance---simply because the calculations of tables of~$C(\dd)$ up to~$g=40$ using the 
recursion~\eqref{DNrec} were done using the software package GP-PARI, which happens to order 
partitions in the way just described---we noticed empirically the following

\begin{conj}\label{conjmonotoncity}
The function $\dd\mapsto C(\dd)$ from partitions of~$g-1$ to~$\Q$
is strictly monotone increasing with respect to the above ordering for every~$g$.
\end{conj}
\noindent To make this property more visible, we have given the numbers $C(\dd)$ in Table~\ref{tablenormalizednumbers} 
both as rational numbers and as real numbers to 6 significant digits.  For ease of reading,
we have also listed the smallest common denominator~$D=D_g$ of these numbers for each~$g$
and then tabulated the integers $D\,C(\dd)$ in the last column.

From the numerical tables we see a different property: the values of the normalized BGW numbers for a fixed~$g$
are very close to each other, e.g. the minimum and maximum values for $g=40$ given in~\eqref{exampleg40}
differ by less than a third of a percent. In view of the nesting property, we can concentrate on only the two
values $C(g-1)$ and~$C(1^{g-1})$, and indeed we can verify that these two numbers are close to each other for all~$g$.  
On one hand, the value of $C(g-1)$ is given by the explicit formula~\cite{DYZ, BR}
\beq\label{smallest}
C(g-1) \= \frac{g}{4^{2g-1}}\thin\binom{2g-1}g^2 \= \frac{(2g-1)!!^3}{2^{g+1} \, (2g-1)!\,g!} \,,\quad g\ge1\,,
\eeq
which by Stirling's formula has the asymptotics
\beq \label{asymsmallest}
C(g-1) \;\sim\; \frac1\pi\,\Bigl(1\m\frac1{4g}\+\frac1{32g^2} 
\+ \frac1{128g^3} \m \frac5{2048g^4}\+ \cdots\Bigr) \,, \quad g\to\infty\,.  
\eeq
On the other hand, as we will see in Section~\ref{secPainleve},  the value of $C(1^{g-1})$ is given by
\beq\label{biggest}
C(1^{g-1}) \=  \,\frac{ 3^{g-1}\thin (g-1)!}{(3g-2)!}\, y_g\,,
\eeq
where the $y_g$ are defined by requiring that the generating series  
\beq\label{defY0602}
Y \: \sum_{g\ge1} y_g\, X^{1-3g} \=  \frac1{4 \, X^2} \+ \frac{9}{4\,X^5} \+ \frac{1323}{16\,X^8}
\+ \frac{108315}{16 \, X^{11}} \+ \frac{62737623}{64 \, X^{14}} \+ \cdots 
\eeq
satisfies the following third-order nonlinear ODE:
\beq\label{P34equation}
Y''' \+  6\,Y\,Y' \,-\, 2\, Y \,-\, X\, Y' \= 0 \,, \qquad ' \= \frac{d}{dX}\,.
\eeq
This equation can be referred to as the {\it Painlev\'e XXXIV equation} (cf.~\cite{BR,CJP, FA, Ince}).
From~\eqref{defY0602} and~\eqref{P34equation} we know that
the coefficients $y_g$ satisfy the recursion 
\beq\label{recyg}
y_g \= (3g-2)(3g-4) \, y_{g-1} \+ \frac{2}{g-1} \,\sum_{h=1}^{g-1}(3h-1) y_h \, y_{g-h} \qquad(g\ge2) 
\eeq
 with the initial value $y_1=1/4$, and from this one can obtain the large $g$ asymptotics
\beq\label{asymptoticsvg}
y_g \;\sim\; A\,\frac{(3g-2)!}{3^{g-1}\,(g-1)!}\,\Bigl(1 \m \frac1{6g} \m \frac{7}{72g^2}
\m \frac{41}{432g^3}\m \frac{1789}{10368g^4} \+  \cdots \Bigr)\,,
\eeq
where $A$ is some positive constant, which by a numerical computation
can be guessed to be $1/\pi$. 
The theoretical determination of this constant is difficult. With the help of the relationship between the Painlev\'e XXXIV equation and the 
Painlev\'e II equation~\cite{CJP, FA} and a method given in~\cite{CMZ},
one can obtain, by employing an appropriate limit of a deep result of Its--Kapaev 
(\cite{IK}\footnote{We thank Lun Zhang for pointing out the reference~\cite{IK}.} and Chapter~11 of~\cite{FIKN}), that $A=1/\pi$. 
Theorem~\ref{thmuniformasymptotics} below will lead to an independent proof of of this evaluation.
Now using~\eqref{biggest} we get 
\beq \label{asymptoticsC1g}
C(1^{g-1}) \;\sim\;   \frac1\pi\,\Bigl(1\m\frac1{6 \, g}\m\frac7{72 \, g^2} 
\m \frac{41}{432 \, g^3} \m \frac{1789}{10368 \, g^4}\+ \cdots\Bigr)\,, \quad g\to\infty\,.
\eeq

We now note that Conjecture~\ref{conjnesting} together with formulas \eqref{asymsmallest} and
  \eqref{asymptoticsC1g} implies that
\beq\label{C(d)-1/pi}
 \frac 1\pi \,-\, \frac1{4\pi\thin g(\dd)}\+ O\Bigl(\frac1{g(\dd)^2} \Bigr) \;\; \le \;\; C(\dd) 
  \;\; \le  \;\; \frac 1\pi  \,-\, \frac1{6\pi\thin g(\dd)} \+ O\Bigl(\frac1{g(\dd)^2} \Bigr) 
\eeq
as $g(\dd)$ tends to infinity. The following theorem, which will be proved in Section~\ref{secproofofuniformleadingasymp},
gives a slight weakening of this, with an unspecified (though effective) O-constant.  

\begin{thm}\label{thmuniformasymptotics} 
For $\dd\in  (\mathbb{Z}_{\ge0})^n$, we have
\beq
\label{STRONG} C(\dd) \= \dfrac1\pi \+ \text O\thin\Bigl(\dfrac1 {g(\dd)}\Bigr)
\eeq
uniformly as the genus $g(\dd)=|\mathbf{d}|+1$ goes to~$\infty$. 
\end{thm}
\noindent 
Explicitly, this theorem 
says that there exists an absolute constant~$K$ such that
\beq \Bigl|C(\dd) - \frac{1}{\pi} \Bigr| \,\leq \,  \frac{K}{g(\dd)}  \qquad \text{\it for~all}\;\; \dd\in (\mathbb{Z}_{\ge0})^n\, . \label{Cdbound}\eeq
We note that Eynard {\it et al}~\cite{EGGGL} 
proved a version of Theorem~\ref{thmuniformasymptotics} for fixed $n$  
but our result is uniform
(and hence includes
the value of~$A$ in~\eqref{asymptoticsvg}, which does not follow from~\cite{EGGGL} since the
value of~$n=g-1$ grows with~$g$).

The monotonicity conjecture (Conjecture~\ref{conjmonotoncity}) says in particular that the normalized numbers  
$C(\dd)$ with partitions $\dd$ of a fixed length are smaller than those of greater length,
so if $I_{g,n}$ denotes the smallest interval containing the normalized BGW numbers 
of length~$n$ for a given~$g$, then $I_{g,n}$ lies strictly to the left of $I_{g,n+1}$.
The numerical data (up to $g=40$) shows that much more is true.  The following picture shows all $31185$ normalized BGW numbers with $g=40$.
\begin{figure}[htbp]
	\centering
	\includegraphics[scale=0.5]{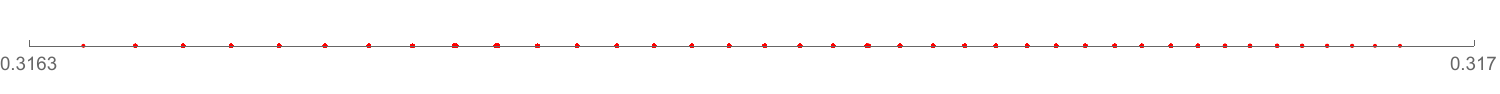}
\end{figure}
At this resolution, what one sees are just 39 intervals that look like points,
meaning that the normalized BGW numbers of length~$n$ are much closer to each other
than to those of length~$n-1$ or~$n+1$. More precisely, 
we make the observation that each interval $I_{g,n}$ has length $O(g^{-3})$, even though the gaps between the 
intervals have an average length of the order $O(g^{-2})$ (because these $g-1$ intervals lie in 
an interval of total length $O(1/g)$ by Theorem~\ref{thmuniformasymptotics}).  
A conjectural statement giving a much more precise result
is stated 
at the end of Section~\ref{secproofofpolynomiality}.

From the numerical data  we observe for some small values of~$n$ that  $C(\dd)/C(g(\dd)-1)$ is a rational
function of $d_n$ if $d_1,\dots,d_{n-1}\ge1$ are fixed.  In fact this is always true, as we will prove in Section~\ref{secallbutonepartfixed}. 
Equation~\eqref{asymsmallest} and Theorem~\ref{thmuniformasymptotics} then imply that $C(\dd)$ for $\dd=(\dd',d_n)$
with $\dd'=(d_1,\dots,d_{n-1})$ fixed has an  
asymptotic expansion of the form
\begin{align}\label{Cdexpansioning}
C(\dd) \,\sim\, \frac{1}{\pi} \, \sum_{k=0}^{\infty}\frac{A_k(\dd')}{g(\dd)^k}\,, \qquad d_n\to\infty\,,
\end{align}
where the $A_k(\dd') $ are rational numbers with $A_0(\dd')=1$. 
Now looking at the explicit formulas for small $n$ (see Section~\ref{secallbutonepartfixed}), we find that 
as $d_1,\dots,d_{n-1}$ grow, the asymptotic expansion of~$C(\dd)$ stabilizes to a well-defined power series in $1/g(\dd)$
depending on~$n$.  And then we discover that if we rewrite them as power series in $1/X(\dd)$, where $X(\dd)=2g-2+n$ as usual, 
we get a power series independent of~$n$ and beginning
\beq 
\frac1\pi\;\Bigl(1 - \frac1{2X} + \frac5{8X^2}-\frac{11}{16 X^3} + \frac{83}{128X^4}-  \frac{143}{256X^5} +\cdots \Bigr)\,.\label{gammaexpansion}
\eeq
We can recognize this power series as the large-$X$ expansion of the function
\begin{align}\label{defgamma}
\gamma(X) \= \frac{\Gamma\bigl(\frac{X}2+1\bigr)^2}{\pi \, \Gamma\bigl(\frac{X+1}2\bigr) \, \Gamma\bigl(\frac{X+3}{2}\bigr)}\,.
\end{align}
We now find that the difference of $C(\dd)$ and $\g(X(\dd))$ is of the order 
$O(X(\dd)^{-2\min\{d_i\} - 2})$, and also that %$C(\dd)$
$C(\dd)\leq \gamma(X(\dd))$ in all cases, with strict inequality unless $n=1$.

We also find that sometimes   
two normalized BGW numbers with the same $g$ and~$n$ are extremely close 
to each other, a numerical example being given the two numbers
\beq \begin{split}
 &C(1,18,20)\approx 0.3163749000332518760707893046\thin, \\ 
	&C(1,19,19) \approx 0.3163749000332518760707893073\thin, \end{split} \label{twonumbers}
\eeq
which agree to 26 significant digits. These phenomena and many others of the same kind will be discussed in 
Sections~\ref{secallbutonepartfixed},\ref{secproofofpolynomiality}, \ref{sectwopoint} and~\ref{secsubleadingasymptotics}.
	
We now turn to the second main theme of this paper, which will shed light on all aspects of the discussion so far.

In the study of Witten's intersection numbers, two of the authors~\cite{GY} discovered, and stated as a conjecture,
that for each~$k$ the coefficient of $1/g^k$ in the large genus asymptotics of  
normalized Witten's intersection numbers is a polynomial of~$n$ and the multiplicities
in the arguments, and also that only the multiplicities of $0, 1, \dots, [3k/2]-1$ are involved.
In the computations for the current paper, we discovered
that the same phenomenon holds also for the normalized BGW numbers $C(\dd)$, 
 now with the $k$th coefficient depending on~$n$ and the multiplicities of $0,1,\dots,[k/2]-1$. 
Both of these conjectural statements
were proved by Eynard {\it et al}~\cite{EGGGL}.

There is a further discovery. We already know that for each~$k$
the coefficient $A_k(\dd')$ in~\eqref{Cdexpansioning}, $\dd'\in(\mathbb{Z}_{\ge1})^{n-1}$, 
is a polynomial $a_k$ of~$n$ and 
the multiplicities. As we will see from Section~\ref{secproofofpolynomiality},
the DVV-type relations for BGW numbers (cf.~Section~\ref{secDN}) imply that 
the power series $\sum_{k=0}^{\infty}a_k/g^k$ is unchanged by 
$(g\to g-1, \, n\to n+2)$.
So, if we write this power series in terms of~$X^{-1}$
instead of~$g^{-1}$ with $X=2g-2+n$, then the coefficients are polynomials of the multiplicities
of $1,\dots, [k/2]-1$, independent of~$n$.
Namely,
\begin{align}\label{Cdexpansionnew}
C(\dd) \,\sim\, \frac{1}{\pi} \, \sum_{k=0}^{\infty}\frac{c_k(p_1(\dd'),\dots,p_{[k/2]-1}(\dd'))}{X(\dd)^k}\,, \qquad d_n\to\infty\,,
\end{align}
for some polynomials $c_k$, where $p_r(\dd')$ denotes the multiplicity of~$r$ in~$\dd'$. The first few polynomials $c_k$ are given by 
$$c_0\,=\,1\,, \quad c_1 \,=\, - \frac 1 {2}\,, \quad 
c_2 \,=\, \frac 5 {8}\,, \quad c_3 \,=\,  - \frac {11} {16} \,, \quad c_4 \,=\, \frac{83}{128}-\frac{27}{8}\,p_1\,. $$
The constant terms of~$c_k$ agree (necessarily) with the coefficients of~$\g(X)$ (cf.~\eqref{gammaexpansion}). 
This makes it very natural to introduce the renormalized BGW numbers $\widehat{C}(\dd)$ by 
\begin{align}
\widehat{C}({\bf d}) \: \frac{C({\bf d})}{\gamma(X(\dd))} \, , \quad \dd\in (\mathbb{Z}_{\ge1})^{n} \,. \label{normalizehatc61}
\end{align}
\begin{thm}\label{thmpoly} 
For any fixed $n$ and fixed $\dd'\in(\mathbb{Z}_{\ge1})^{n-1}$, the numbers $\widehat{C}(\dd)$ satisfy
\begin{align}\label{Chatdexpansionnew}
\widehat{C}(\dd) \,\sim\, \sum_{k=0}^{\infty}\frac{\widehat{c}_k(p_1(\dd'),p_2(\dd'),\dots)}{X(\dd)^k}\,, \qquad d_n\to\infty\,,
\end{align}
where $\dd=(\dd',d_n)$, $\widehat{c}_k$ are universal polynomials of $p_1,p_2,\dots$ having rational coefficients, 
with $\widehat{c}_0\equiv1$ and $\widehat{c}_k|_{p_b\equiv 0}=0$ $(k\ge1)$.
Moreover, under the degree assignments 
\begin{equation}\label{degreeassign}
\deg \, p_d \= 2d+1\quad (d\geq 1)\,,
\end{equation}
the polynomials $\widehat{c}_k$, $k\ge1$, satisfy the degree estimates
\begin{equation}\label{degreeest}
\deg \, \widehat{c}_k \,\leq\, k-1\,.
\end{equation} 
\end{thm}

\noindent We will give a proof of this theorem, independent of~\cite{EGGGL}, in Section~\ref{secproofofpolynomiality}.

From~\eqref{degreeest}, we know that $\widehat{c}_k$ does not depend on~$p_{d}$ with $d\ge (k-1)/2$.
In particular, $\widehat{c}_k=0$ for $k=1,2,3$.
We list a few more $\widehat{c}_k$ below:
$$
\widehat{c}_4 \,=\, -\frac{27}{8}\,p_1 \,, \quad \widehat{c}_5 \,=\, - \frac{27}{4}\,p_1\,, \quad \widehat{c}_6 \,=\, -\frac{45}4 \, p_1 - \frac{1125}{16} \, p_2\,.
$$
Several more coefficients for both $c_k$ and $\widehat{c}_k$ are given in Table~\ref{tablepoly} of Section~\ref{secproofofpolynomiality}.
We are also able to give explicit expressions for all $\widehat{c}_k|_{p_b=\delta_{b,d}}$; see Section~\ref{sectwopoint}.

By Theorem~\ref{thmpoly} and the result~\cite[Theorem~4.3]{EGGGL} 
of Eynard {\it et al}, we also know that 
for any fixed $L\ge1$ and fixed $n\ge1$, and for $\dd=(d_1,\dots,d_n)\in(\ZZ_{\ge1})^n$,
\begin{align}\label{Chatd1dnuniformindwithout0}
	\widehat{C}(\dd) \= \sum_{k=0}^{L-1} \frac{\widehat{c}_{k}(p_1(\dd),p_2(\dd),\dots)}{X(\mathbf{d})^k} + O\biggl(\frac{1}{X(\mathbf{d})^{L}}\biggr)\,,  \qquad g(\mathbf{d})\to\infty\,,
\end{align}
where the implied O-constant only depends on~$n$ and~$L$.

Based on an algorithm given in~\cite{DYGUE, DYP1} (see~equations \eqref{defMmd}--\eqref{ijbk} of Section~\ref{secexplicitform})
 we have computed $\widehat{C}(d^n)$ for genera far bigger than~40, from which we also see the phenomenon
 that $\widehat{C}(d^n)$ rapidly tends to~1 as~$d\to\infty$. 
(For example, $1-\widehat{C}(100^{10})$ is roughly $1.8\times10^{-285}$.) This together with Theorem~\ref{thmpoly} 
leads us to the discovery of the conjectural asymptotic formula 
\begin{align}\label{dddasympleading}
1 \,-\, \widehat{C}(d^n) \,\sim\, \Bigl(\frac12\Bigr)^{\delta_{n,2}}\,
\sqrt{\frac{4 \,(n-1)}{\pi\,n\,d}}\, \biggl(\frac{(n-1)^{n-1}}{n^n}\biggr)^{2d+1}\,,\quad
{\rm as}~d\to\infty\,.
\end{align}
More details and generalizations of this will be discussed in Section~\ref{secsubleadingasymptotics}.

We end this section
by presenting two applications of Theorem~\ref{thmuniformasymptotics}. 

The first one is an application for the Painlev\'e II hierarchy. 
Following~\cite{BDY,CJP, Dickey,Magri}, define a sequence of 
 polynomials $m_d=m_d(u_0,u_1,u_2,\dots,u_{2d})$, $d\ge0$, by means of generating series 
as follows: 
\begin{align}\label{defma}
b \, \p^2(b) \,-\, \frac{1}{2} \, \p(b)^2 \,- \, 2\, (\lambda-2u_0) \, b^2 \= -2 \,\lambda \,,
\end{align}
where $\p := \sum_{i\ge0} u_{i+1} \p / \p u_i$, and 
\begin{align}\label{b(lambda)expansion}
b(\lambda) \= 1\+\sum_{d\ge0}\frac{(2d+1)!! \, m_d }{\lambda^{d+1}}.
\end{align} 
The first few $m_d$ are $m_0=u_0$, $m_1=\frac12 u_0^2 + \frac1{12} u_2$.
By the {\it Painlev\'e II hierarchy} we mean the following family of ODEs:
\begin{align}\label{P2hier}
2^{2d-1} (2d-1)!! \, (\partial_X
+2V) \, \Bigl(m_{d-1}\Bigl(\tfrac{V_X-V^2}2, \tfrac{(V_X-V^2)_X}2, \dots\Bigr)\Bigr)-VX-\alpha_d\=0\,,
\end{align} 
where $d\ge1$ and $\alpha_d$ are constants. We will focus on the case when $\alpha_d=1/2$, $d\ge1$.
In this case, it can be shown that, for each $d\ge1$, there exists a unique formal solution $V(X)$ to~\eqref{P2hier} 
of the form
\begin{align}\label{solP2hier}
V(X) \= -\sum_{n=0}^{\infty}\frac{v_{d,n}}{X^{(2d+1)n+1}}\,,  \qquad v_{d,n}\in \mathbb{C} \,, \quad v_{d,0}=\frac12\,.
\end{align} 
In Section~\ref{secPainleve} we will use Theorem~\ref{thmuniformasymptotics} to prove the following theorem.
\begin{thm}\label{thmPainleveapp}
For each $d\ge1$, the coefficients $v_{d,n}$ of the formal solution $V(X)$ to the $d$th member of the 
Painlev\'e II  hierarchy have the following asymptotics:
\begin{align}
v_{d,n} \; \sim  \; \frac1\pi \, \frac{((2d+1)n-1)!}{(2d+1)^{n-1} \thin (n-1)!}\,, \quad n\to\infty\,.\label{asymvak}
\end{align}
\end{thm}
We note that for the particular case when $d=1$ the above theorem was essentially proved in \cite{IK} and~\cite{FIKN} by 
using a deep Riemann--Hilbert analysis, and we now achieve a new proof.
As far as we know, the cases with $d\ge2$ are new.

The second application that we will present 
is to use Theorem~\ref{thmuniformasymptotics} to study the large genus asymptotics of the more general integrals, which we 
call the {\it BGW-kappa numbers}, where the Theta-class is coupled with powers of $\kappa_1$-class as well as 
psi-classes
\beq\label{defkappatheta}
\int_{\overline{\mathcal{M}}_{g,n}} \,  \psi_1^{d_1}\cdots\psi_n^{d_n} \, \Theta_{g,n} \, \kappa_1^m \;=:\; \langle\kappa_1^m\prod_{j=1}^n\tau_{d_j}\rangle_g^{\Theta}\,.
\eeq
By the degree-dimension matching, these numbers vanish unless $m+d_1+\cdots+d_n=g-1$. 
A small table of these BGW-kappa numbers is provided in Section~\ref{seckappaclass}.

Like the numbers $C(\mathbf{d})$, we introduce the {\it normalized BGW-kappa numbers} as follows:
\begin{align}\label{defCkappa}
C(m;\dd) \,:=\, \frac{3^m\, 2^{2g-1}  \prod_{j=1}^n (2d_j+1)!!}{(X(m;\mathbf{d})-1)!}\, \langle\kappa_1^m\tau_{d_1}\cdots\tau_{d_n}\rangle_{g}^{\Theta}\,, \quad X(m;\dd) \,:=\, X(\dd)+3m\,.
\end{align}
Obviously, $C(0;\mathbf{d})=C(\mathbf{d})$. 
In Section~\ref{seckappaclass} we will use Theorem~\ref{thmuniformasymptotics} to prove  the following
\begin{prop}\label{corkappa}
For any fixed $m\ge0$, 
there exists a constant $K(m)$ such that
\begin{align}
	\bigg| C(m;\mathbf{d}) -\frac1\pi\bigg| \;\leq\; \frac{K(m)}{g(m;\dd)} \qquad \text{\it for~all}\;\; \dd\in (\mathbb{Z}_{\ge0})^n\,,
\end{align}
where $g(m;\dd)=g(\dd)+m=|\dd|+m+1$.
\end{prop}

We also give in Section~\ref{seckappaclass} an application of Theorem~\ref{thmpoly} to BGW-kappa numbers.

The paper is organized as follows.  
In Section~\ref{secgeneral} we review a recursive definition of 
BGW numbers as well as an explicit formula for their $n$-point generating series. 
In Section~\ref{secpartitionfixedlength} we give closed formulas for BGW numbers. 
In Section~\ref{secallbutonepartfixed} we present several results on structures for 
BGW numbers. 
In Section~\ref{secproofofuniformleadingasymp} we prove Theorem~\ref{thmuniformasymptotics},  
and in Section~\ref{secproofofpolynomiality} we prove Theorem~\ref{thmpoly}. Further asymptotic formulas 
and conjectural subexponential asymptotics are given in Sections~\ref{sectwopoint}, \ref{secsubleadingasymptotics}, respectively.
In Sections~\ref{secPainleve}, \ref{seckappaclass} we present applications of the main theorems. 

\smallskip

\noindent {\bf Acknowledgements} 
The work was partially supported by NSFC No.~12371254, the CAS No.~YSBR-032, 
the National Key R and D Program of China 2020YFA0713100, and the China Scholarship Council.
Parts of the work of J.G., P.N. and D.Y. were done during their visits in MPIM, Bonn; they thank MPIM for excellent working
conditions.

\section{Review of general theory of BGW numbers}
\label{secgeneral}

The definitions for BGW numbers given in Section~\ref{secintro} are not directly calculable since the integrals over the
unitary group or moduli space are not algorithmically defined.  
In this section we review two algorithms
that can be used to effectively compute the BGW numbers and that can and have been implemented on a computer.

\subsection{Recursive definition of the BGW numbers}  \label{secDN} 
It is known from~\cite{DN,MMS} (cf.~also~\cite{A,CGG,Norbury0}) that the partition function $Z$ of the BGW numbers, defined by
\beq\label{partitionfunctiondef}
Z \= \exp \Biggl( \, \sum_{g\ge1} \sum_{n\ge1}\, \frac{1}{n!}\sum_{d_1,\dots,d_n\ge0 \atop d_1+\cdots+d_n=g-1}\langle \tau_{d_1} \dots \tau_{d_n} \rangle^{\Theta}_g \, t_{d_1} \cdots t_{d_n} \Biggr)
\eeq
(cf.~\eqref{ddc}), is a particular tau-function for the celebrated KdV hierarchy.
In particular, $u:=\partial^2 \log Z/\partial_{t_0}^2$ satisfies the KdV hierarchy 
\beq 
u_{t_d}\= \partial_{x}(m_d(u,u_x,u_{xx},\dots,u_{2dx}))\,, \quad d\ge0\,, 
\label{KdV}
\eeq
with $t_0\equiv x$, 
where $m_d$ is defined in~\eqref{defma},~\eqref{b(lambda)expansion}. It is also known that $Z$ satisfies the following infinite set of 
linear equations called the {\it Virasoro constraints}:
\begin{align}\label{LmZ=0}
L_m Z \=0\,,\quad m\ge0\,,
\end{align}
where $L_m$, $m\ge0$, are operators defined by
\begin{align}
&L_m \: -(2m+1)!!\, \frac{\partial}{\partial t_m} +\sum_{d\ge0} \frac{(2d+2m+1)!!}{(2d-1)!!} \, t_d\, \frac{\partial}{\partial t_{d+m}} \nn\\
&\qquad +\frac{1}{2}\sum_{a+b=m-1} (2a+1)!!\, (2b+1)!!\, \frac{\partial^2}{\partial t_a\thin \partial t_b} +\frac{1}8 \delta_{m,0}\,. \label{defL}
\end{align}
See~\cite{A, DN,GN, MMS} (cf.~also~\cite{AC, BR, CGG, DYZ,   Norbury0}).

For $n\ge1$, ${\bf d}=(d_1,\dots,d_n)\in(\ZZ_{\ge0})^n$, it is convenient to denote 
\beq\label{tauU}  
B({\bf d}) \:  \langle\tau_{d_1}\dots\tau_{d_n}\rangle_{g(\dd)}^{\Theta} \, \prod_{j=1}^n(2d_j+1)!! \,,
\eeq
where we recall that $g(\dd)=|\dd|+1$.
Using the $m=0$ case of~\eqref{LmZ=0}, we obtain
\beq B(0,\dd) \= (2g(\dd)+n-2)\, B(\dd)\,, \quad 2g(\dd)-2+n>0\,.\label{BB0}\eeq
In general, a recursion for the BGW numbers that is equivalent to the Virasoro constraints
was derived in~\cite{DN} by Do and the second author of the present paper:
\beq\begin{aligned} \label{DNrec}
&B(d,\dd) \= \sum_{i=1}^n (2d_i+1)\,B(d_1,\dots,d_i+d,\dots,d_n) \\
&\quad\+ \frac12\sum_{a+b=d-1}\Bigl(B(a,b,\dd) 
\+ \sum_{I\sqcup J=\{1,\dots,n\}} B\bigl(a,\{d_i\}_{i\in I}\bigr)\,B\bigl(b,\{d_j\}_{j\in J}\bigr)\Bigr)\,,
\end{aligned} 
\eeq
where $d\ge0$.
We refer to~\eqref{DNrec} as 
the DVV-type relation (here ``DVV" stands for Dijkgraaf--Verlinde--Verlinde), 
because it is analogous to the DVV relation for 
Witten's intersection numbers~\cite{DVV}.
It is also closely related to the 
topological recursion~\cite{DN, EO}.
Note that originally the DVV-type relation for BGW numbers was written in another normalization, denoted $U_{g,n}(2d_1+1,\dots,2d_n+1)$~\cite{DN}, which is related to $B({\bf d})$ by 
$$U_{g,n}(2d_1+1,\dots,2d_n+1) \= B({\bf d}) \; /\; \prod_{j=1}^n (2d_j+1) \,.$$

By induction on the sum $\sum_{i=1}^n(2d_i+1)=2g(\dd)+n-2$, we know that 
the numbers $B(\dd)$
can be uniquely determined by~\eqref{DNrec} along with the initial value 
\beq B(0)\=\frac{1}{8}\label{B0} \eeq
(see \eqref{BGWnumbers}, \eqref{tauU}).
However, 
it is not at all obvious from~\eqref{DNrec} that $B$ is symmetric in its arguments. In other words, if we force this
symmetry by defining $B$ as a function on unordered multisets, then \eqref{DNrec} is an overdetermined system
because we can choose any of the $n+1$ arguments of $B(d_1,\dots,d_{n+1})$ as the ``$d$" of~\eqref{DNrec} and it is 
non-trivial that the right-hand side will be independent of this choice.

Using the DVV-type relation~\eqref{DNrec}, we can in principle
compute the numbers $C(\dd)$ for partitions~$\dd=(d_1,\dots,d_n)$ of~$g-1$ for any $g$, 
and we have done so for all $g$ up to~40.  

\subsection{Explicit formulas of $n$-point generating series}  \label{secexplicitform}
Following~\cite{BDY, BR, DYZ}, consider the following $n$-point generating series of the BGW numbers:
\begin{equation}\label{defFn}
F_n(\lambda_1,\dots,\lambda_n) \=
\sum_{d_1,\dots,d_n\ge0}\frac{B(d_1,\dots,d_n)}{\lambda_1^{d_1+1}\cdots\lambda_n^{d_n+1}}\,.
\end{equation}
Using the matrix-resolvent method~\cite{BDY}, an explicit formula for $F_n(\lambda_1,\dots,\lambda_n)$ 
was obtained in~\cite{DYZ} (see~\cite{BR} for a different proof)  
\begin{align}
 & F_1(\lambda) \=\sum_{d\ge0} \, \frac{(2d+1)!!^3}{8^{d+1} (d+1)! (2d+1)} \, \frac{1}{\lambda^{d+1}}\,, \label{f1}\\
 & F_n(\lambda_1,\dots,\lambda_n) 
  \= -\frac1n\sum_{\sigma\in S_n}\frac{\mathrm{tr}(M(\lambda_{\sigma(1)})\cdots M(\lambda_{\sigma(n)}))}
   {\prod_{i=1}^{n}(\lambda_{\sigma(i+1)}-\lambda_{\sigma(i)})}-\delta_{n,2}\frac{\lambda_1+\lambda_2}{(\lambda_1-\lambda_2)^2}\,, \quad n\ge2\,,\label{fn}
\end{align}
where 
\beq\label{MatRes} 
 M(\l)\:\sum_{k\ge -1} \biggl(\frac{(2k-1)!!}{2^k}\biggr)^3\,
\begin{pmatrix} k\,(k+1) & k+1 \\
  -\frac{8k^3+12k^2+4k+1}{8}  & -k\,(k+1) \end{pmatrix}\,\frac{\l^{-k}}{(k+1)!}\,, 
\eeq
with the usual conventions $(-1)!!=1$ and $(-3)!!=-1$.

There is a useful variant of formulas like~\eqref{fn} given by 
Dubrovin and the third author of the present paper (see~\cite[Proposition 3.2.3]{DYGUE}).
A special case of the variant 
gives the numbers $\langle\tau_{d_1}\dots\tau_{d_n}\rangle_g^{\Theta}$ efficiently when 
all but at most two of the $d$'s are equal. Indeed, define a sequence of traceless $2\times2$ matrix-valued 
functions $M_{m,d}(\l)$ ($m\ge0$) by setting $M_{0,d}(\l)=M(\l)$ as above (independent of~$d$) and then inductively defining
\beq M_{m,d}(\l)  \=  \frac1m\,\sum_{i+j=m-1} \bigl[\bigl(\l^dM_{i,d}(\l)\bigr)^-,\,M_{j,d}(\l)\bigr]  \label{defMmd}\eeq
for $m\ge1$, where $A(\l)^-$ denotes the sum of the terms with strictly negative exponents of a Laurent 
series~$A(\l)$ in~$1/\l$. Then we have the following generating function formula:
\beq \label{ijbk}
\frac{(\l_1-\l_2)^2}{m!}\, \sum_{a,\thin b\geq 0} \frac{B(a,b,d^m)}{\l_1^{a+1}\l_2^{b+1}}  
  \= \sum_{k=0}^m\tr\bigl( M_{k,d}(\l_1)\,M_{m-k,d}(\l_2)\bigr)\;-\,\delta_{m,0}\;. 
\eeq
This formula allows us to calculate all of the numbers $C(a,b,d^{n-2})$, and in particular
the numbers $C(d^n)$, quite efficiently even when the genus is large, in which case the
recursive formula~\eqref{DNrec} would be useless because it requires one to have computed and stored the BGW numbers
for all smaller genera.  Using it, we computed the rational numbers $C(d^n)$ for $1\le d\le 100$ and~$1\le n\le 10$.
This computation took about 20 hours on a relatively fast
desktop computer, which sounds like a lot until one realizes that, for example, the numerator and denominator of
the rational number $C(100^{10})$ each has 3020 digits.

\section{Exact formulas for $n$-point BGW numbers} \label{secpartitionfixedlength}
For $n\ge2$, expanding the right-hand side of~\eqref{fn} in the region $|\lambda_1|>\dots>|\lambda_n|\gg0$, one gets a formula for $n$-point BGW numbers which is similar to a  formula for Witten's intersection numbers given in~\cite{GY}. For $\sigma\in S_n$, introduce the notation 
\begin{align*}
S_{\sigma}^+ \=\left\{1\leq r\leq n \,\big|\, \sigma(r+1)>\sigma(r)\right\},\quad 	S_{\sigma}^- \=\left\{1\leq r\leq n \,\big|\, \sigma(r+1)<\sigma(r)\right\}.
\end{align*}
Here $\sigma$ is considered to be cyclic, so that we have the convention $\sigma(n+1)=\sigma(1)$.
For $k_1,\dots,k_n\ge-1$, introduce the notation 
\begin{align} \label{defak1kn}
a_{k_1,\dots,k_n} \: \tr(A_{k_1}\cdots A_{k_n})\,, 
\end{align}
where $A_k := f(k) R(k)$, $k\ge-1$, with 
\beq\label{deffkRk}
f(k) \: \frac{(2k-1)!!^3}{2^{3k}(k+1)!}\,,\quad R(k) \: \begin{pmatrix} k(k+1) &k+1\\ -\frac{8k^3+12k^2+4k+1}{8} &-k(k+1)\end{pmatrix}\;, %\label{deffkRk} 
\eeq 
and we make the convention that $a_{k_1,\dots,k_n}=0$ if any of the $k_i$ is less than or equal to~$-2$.
Note that $A_k$ is just the coefficient of $\lambda^{-k}$ in the Laurent series $M(\lambda)$ defined in~\eqref{MatRes}.

\begin{prop}\label{prop1}
For $n\ge 2$ and ${\bf d}=(d_1,\dots,d_n)\in (\ZZ_{\ge 0})^n$, we have
\begin{align}\label{formulaBGWeq1}
B({\bf d}) \= \sum_{\substack{\sigma\in S_{n} \\ \sigma(n)=n}} (-1)^{|S_{\sigma}^+|+1}
\sum_{\substack{\underline{J}\in (\mathbb{Z}+\frac12)^n \\ \left\{1\leq q\leq n\mid J_q>0\right\}=S_{\sigma}^{+}}} 
a_{d_{\sigma(1)}+J_1-J_n,\dots,d_{\sigma(n)}+J_n-J_{n-1}} \,.
\end{align}
\end{prop}
\begin{proof} 
In the region $|\lambda_1|>\dots>|\lambda_n|\gg0$, we have the Laurent expansion
\begin{align}\label{Pexpansion}
\prod_{q=1}^n \frac1{\lambda_{\sigma(q+1)}-\lambda_{\sigma(q)}} \= (-1)^{|S_{\sigma}^+|} \sum_{j_1,\dots,j_n\geq0} \prod_{q=1}^n \lambda_{\sigma(q)}^{J_{\sigma,q}(j_q)-J_{\sigma,q-1}(j_{q-1})-1},
\end{align}
where for $q=1,\dots,n$,
\begin{align}\label{defJsigmaq}
J_{\sigma,q}(j) := 
\left\{\begin{array}{ll}
-j-1,  &\sigma(q)<\sigma(q+1), \\
j,   &\sigma(q)>\sigma(q+1).
\end{array}\right.
\end{align}
Expanding both sides of~\eqref{fn} and using~\eqref{Pexpansion}, we get
\begin{align}\label{formulaBGWeq3}
B(\dd) \= \sum_{\substack{\sigma\in S_{n} \\ \sigma(n)=n}} (-1)^{|S_{\sigma}^+|+1} \sum_{\underline{j}\in(\mathbb{Z}_{\ge0})^n} 
a_{d_{\sigma(1)}+J_{\sigma,1}(j_1)-J_{\sigma,n}(j_n),\dots,d_{\sigma(n)}+J_{\sigma,n}(j_n)-J_{\sigma,n-1}(j_{n-1})}\,.
\end{align}
For each $\sigma\in S_n$ with $\sigma(n)=n$, by changing the variable $J_{q}=\frac{1}{2}+J_{\sigma,q}(j_q)$, $q=1,\dots,n$, we obtain formula~\eqref{formulaBGWeq1}.
\end{proof}

Proposition~\ref{prop1} could be rewritten in a more elegant way as follows:
\begin{prop}\label{prop2}
For $n\ge 2$ and ${\bf d}=(d_1,\dots,d_n)\in (\ZZ_{\ge 0})^n$, we have
\begin{align}\label{formulaBGWeq2}
B({\bf d}) \= \sum_{\substack{\sigma\in S_{n} \\ \sigma(n)=n}} (-1)^{|S_{\sigma}^{+}|+1} \sum_{\substack{k_1,\dots,k_n\ge-1\\ k_1+\cdots+k_n=d_1+\cdots+d_n}} 
a_{k_1,\dots,k_n}\,\omega_{\dd,\sigma,\mathbf{k}}\,,
\end{align}
where the numbers $\omega_{\mathbf{d},\sigma,\mathbf{k}}$ have the following explicit expression
\begin{align}\label{omegadsigmak}
\omega_{\mathbf{d},\sigma,\mathbf{k}}\=\max\biggl\{0,\,\min_{r\in S_{\sigma}^+}\bigg\{\sum_{q=1}^r(d_{\sigma(q)}-k_q)\bigg\} 
+\min_{r\in S_{\sigma}^-}
\bigg\{\sum_{q=1}^r(k_q-d_{\sigma(q)})\bigg\} \biggl\}.
\end{align}
\end{prop}
\begin{proof} 
For each $\sigma\in S_n$ with $\sigma(n)=n$ we change the variable  $k_q=d_{\sigma(q)}+J_{\sigma,q}(j_q)-J_{\sigma,q-1}(j_{q-1})$, $q=1,\dots,n$, in formula~\eqref{formulaBGWeq3} and we get
\begin{align}\label{formulaBGWeq4}
B(\dd) \= \sum_{\substack{\sigma\in S_{n} \\ \sigma(n)=n}} (-1)^{|S_{\sigma}^{+}|+1} \sum_{\substack{k_1,\dots,k_n\ge-1\\ k_1+\cdots+k_n=d_1+\cdots+d_n}} 
a_{k_1,\dots,k_n}\,\omega_{\dd,\sigma,\mathbf{k}}\,,
\end{align}
where the numbers
$\omega_{\dd,\sigma,\mathbf{k}}$ are the number of solutions $\underline{j}\in\left(\ZZ^{\ge0}\right)^n$ for the linear equations:
\begin{align}\label{K=k}
d_{\sigma(q)}+J_{\sigma,q}(j_q)-J_{\sigma,q-1}(j_{q-1})\=k_q, \quad q=1,\dots,n\,.
\end{align}
Now it suffices to prove that these numbers $\omega_{\dd,\sigma,\mathbf{k}}$ have the expression~\eqref{omegadsigmak}.
Indeed, by using~\eqref{defJsigmaq}, equations~\eqref{K=k} can be solved in terms of $j_n$ by
\begin{align}\label{jr}
j_r\=\left\{\begin{array}{ll}
-j_n-1+\sum_{q=1}^r(d_{\sigma(q)}-k_q),\quad &r\in S_{\sigma}^+, \\
j_n+\sum_{q=1}^{r}(k_q-d_{\sigma(q)}),\quad &r\in S_{\sigma}^-.
\end{array}\right.
\end{align}
Here we use $\sigma(n)=n$ to obtain $J_{\sigma,n}(j_n)=j_n$.
Therefore, $\omega_{\dd,\sigma,\mathbf{k}}$ is equal to the number of $j_n\in\mathbb{Z}_{\ge0}$ such that $j_r\ge0$ in~\eqref{jr} for all $r=1,\dots,n-1$, and hence is equal to the right-hand side of~\eqref{omegadsigmak}. This finishes the proof.
\end{proof}
We note that when $n\ge3$ and some of $d_j$ are less than zero, then formula~\eqref{formulaBGWeq1} or formula~\eqref{formulaBGWeq2} still holds true (where both sides are~0), since both sides are the coefficients of $\bl_1^{-d_1-1}\cdots\bl_n^{-d_n-1}$ in the power series $F_n(\lambda_1,\dots,\lambda_n)$ defined in~\eqref{defFn}.

Let us give some examples for Proposition~\ref{prop2}. First we introduce the following notation: for $n\ge1$ and $\mathbf{e}\in\mathbb{Z}^n$, write
\begin{align}
	M(\mathbf{e}) \=\max\bigl\{0,\min_{1\leq i\leq n}\{e_i\}\bigr\}
\end{align}
which can be written in terms of a generating function by
\begin{align}\label{generatingM}
\sum_{e_1,\dots,e_n\ge0}x_1^{e_1-1}\cdots x_n^{e_n-1}\,M(e_1,\dots,e_n)\=\frac{1}{(1-x_1\cdots x_n)\prod_{i=1}^n(1-x_i)}\,.
\end{align}
For $n=2$, Proposition~\ref{prop2} reads that
\begin{align}\label{twopoint1}
B(d_1,d_2) \= \sum_{k_1+k_2=d_1+d_2} M(d_1-k_1) \, a_{k_1,k_2} \,,
\end{align}
where $a_{k_1,k_2}$ can be explicitly given as follows:
\begin{align}
&a_{k_1,k_2}\=-f(k_1)f(k_2)\bigl(\bigl((k_1-k_2)^2+\tfrac{k_1+k_2}{2}\bigr)(k_1+1)(k_2+1)-\tfrac{k_1+k_2+2}8\bigr)\,,\nn
\end{align}
with $f(k)$ defined in~\eqref{deffkRk}. 
Actually, we also have a simpler formula for $B(d_1,d_2)$:
\begin{align}\label{TwoPoints}
B(d_1,d_2)\=\frac{1}{g}\; \sum_{h=0}^{d_1} \, (g-2h)\, F_{h}\, F_{g-h}\,, \qquad \text{where\ } F_h\:\frac{(2h-1)!!^3}{2^{3h}\,h!}\,.
\end{align}
The equivalence of \eqref{twopoint1} and~\eqref{TwoPoints} can be proved as in~\cite{Guo}.  Since $\sum_{h=0}^g(g-2h)\,F_h\,F_{g-h}$
vanishes by antisymmetry, we see the RHS of~\eqref{TwoPoints} is indeed symmetric in~$d_1$ and~$d_2$.
We also note that formula~\eqref{TwoPoints} also holds for $d_1=-1$, $d_2\ge1$ and $d_2=-1$, $d_1\ge1$ (where both sides are 0).  
With the normalization $C(\mathbf{d})$ the above formula becomes
\beq  C(d_1,d_2) \= \frac{2^{2g}}{(2g)!} \, \sum_{h=0}^{d_1}(g-2h)\,F_h\,F_{g-h}\,.  \qquad 
%\text{where\ } F_h:=\frac{(2h-1)!!^3}{2^{3h}\,h!}.  
\label{TwoPointsC} 
\eeq
Formula~\eqref{TwoPoints} (or say~\eqref{TwoPointsC}) is independently found in a recent paper~\cite{HLX}.
For $n=3$, Proposition~\ref{prop2} reads
\begin{align}\label{threepoint}
B(d_1,d_2,d_3)\=-2\sum_{k_1+k_2+k_3=d_1+d_2+d_3 } a_{k_1,k_2,k_3} \, M(d_1-k_1,d_1+d_2-k_1-k_2)\,,
\end{align}
where $a_{k_1,k_2,k_3}$ can be given more explicitly by
\beq  a_{k_1,k_2,k_3} \= f(k_1)\, f(k_2)\, f(k_3) (k_1-k_2)(k_2-k_3)(k_3-k_1)\bigl((k_1+1)(k_2+1)(k_3+1)+\tfrac18\bigr)\,. \label{trRk1Rk2Rk3}\eeq 
For $n=4$,  Proposition~\ref{prop2} reads that
\begin{align}
&B(d_1,d_2,d_3,d_4)\=2\sum_{k_1+k_2+k_3+k_4=d_1+d_2+d_3+d_4}a_{k_1,k_2,k_3,k_4} \nn\\
&\quad \times \Bigl(M\bigl(d_1-k_1,d_1+d_2-k_1-k_2,k_4-d_4\bigr) \nn\\
&\qquad -M\bigl(d_1-k_2,d_1+d_2-k_2-k_3,d_1+d_3-k_1-k_2,k_4-d_4\bigr) \nn\\
&\qquad -M\bigl(d_1-k_1,d_2-k_3,k_2-d_3,k_4-d_4\bigr)\Bigr)\,, \label{fourpoint}
\end{align}
where we have used the fact that $a_{k_1,k_2,k_3,k_4}=a_{k_2,k_3,k_4,k_1}=a_{k_1,k_4,k_3,k_2}$.

\section{Rational functions, asymptotics and integrality} \label{secallbutonepartfixed}  
In this section we consider the 
numbers $C(\bl,g-1-|\bl|)$, where $\bl$ is a given fixed partition and we allow $g$ to vary.  

When~$|\bl|=0$, the formula for $C(g-1)$ is known; see~\eqref{smallest}. Based on 
the DVV-type relation~\eqref{DNrec}, 
we can 
 find that for $ |\bl| =1,2,3$ the quotient of $C(\bl,g-1-|\bl|)$ by $C(g-1)$ is a rational function of~$g$. Explicitly, %, the values for $ |\bl| =1,2,3$ being 
\begin{align*}
\frac{C(1,g-2)}{C(g-1)}&\= \frac{g-1}{(2g-1)^3}\,Q_1(g)\,, \\
\frac{C(\bl,g-3)}{C(g-1)}&\= \frac{(g-1)(g-2)}{(2g-1)^3(2g-3)^3}\,Q_{\bl}(g)\,,\quad |\bl|=2\,,\\
\frac{C(\bl,g-4)}{C(g-1)}&\= \frac{(g-1)(g-2)(g-3)}{(2g+1)(2g-1)^3(2g-3)^3(2g-5)^3}\,Q_{\bl}(g) \,, \quad |\bl|=3\,,
\end{align*}
where   % ~$P_{\dd}(g)$
\begin{align*}   
Q_1(g) & \= 8g^2-4g+3 \,,\\
Q_{1,1}(g) &\,=\, 64g^4-192g^3+224g^2-144g+117\,, \nn\\
Q_2(g) &\,=\,64g^4-192g^3+216g^2-108g+\tfrac{135}2\,,   \nn\\
 Q_{1,1,1}(g) &\,=\, 1024g^7-7168g^6+19072g^5-24256g^4 +18832g^3-15520g^2+11418g+14823\,, \nn\\
Q_{1,2}(g) &\,=\, 1024g^7 - 7168g^6 + 18944g^5 - 23040g^4 + 14056g^3 - 6272g^2 + 5411g + \tfrac{16365}2\,, \nn\\
Q_3(g) &\,=\, 1024g^7 - 7168g^6 + 18816g^5 - 21952g^4 + 10360g^3 + 1125g + \tfrac{7875}2 \,. \nn
 \end{align*}

The general situation is described in the following two propositions.
\begin{prop}\label{proprationalabc}
For $n\ge1$ and for a given fixed partition $\bl=(\lambda_1,\dots,\lambda_{n-1})$, we have
\begin{align}\label{eqrational1}
\frac{C(\bl,g-1-|\bl|)}{C(g-1)} \= \frac{\widetilde{Q}_{\bl}(g)}{(2g-2|\bl|-2)_{2|\bl|+n} \, \prod_{i=1}^{|\bl|}(2g-2i+1)^2} \, ,
\end{align}
where $\widetilde{Q}_{\bl}(g)\in \mathbb{Z}[1/2][g]$ is a polynomial, and $(a)_b:=a(a+1)\cdots (a+b-1)$ denotes the ordinary Pochhammer symbol.
\end{prop}
\begin{proof}
Write the left-hand side of~\eqref{eqrational1} as~$\widetilde{\Phi}_{\bl}(g)$. 
In terms of~$\widetilde{\Phi}_{\bl}(g)$, 
recursion~\eqref{DNrec} reads
\begin{align}
&\widetilde{\Phi}_{d,\bl}(g)\=\sum_{j=1}^{n-1}\frac{2\lambda_j+1}{2g+n-2}\,
\widetilde{\Phi}_{\lambda_1,\dots,\lambda_j+d,\dots,\lambda_{n-1}}(g)+\Big(1-\frac{2d+2|\bl|+n-1}{2g+n-2}\Big)\,\widetilde{\Phi}_{\bl}(g) \nn\\
&  +\sum_{\substack{a+b=d-1\\ a,b\ge0}}\Bigg[\frac{8g(g-1)}{(2g+n-2)(2g-1)^2}\,\widetilde{\Phi}_{a,b,\,\bl}(g-1)\nn\\
& +\sum_{I\sqcup J=\{1,\dots,n-1\}}2\,\big(X(a,\bl_{I})-1\big)!\,C(a+|\bl_{I}|)\,
\widetilde{\Phi}_{a,\bl_I}\Big(a+1+|\bl_I|\Big)\, \widetilde{\Phi}_{b,\bl_J}\Big(g-1-a-|\bl_I|\Big)
\nn\\
&\qquad \times\, 
\frac{\bigl(g-a-|\bl_I|\bigr)_{a+1+|\bl_I|}\,
		\bigl(g-1-a-|\bl_I|\bigr)_{a+1+|\bl_{I}|}}{\bigl(2g+n-2-X(a,\bl_I)\bigr)_{X(a,\bl_I)+1}\, \big((g-\frac{1}2-a-|\bl_I|)_{a+1+|\bl_{I}|}\big)^2}\Bigg]. \label{Phirec}
\end{align}
Note that $\widetilde{\Phi}_{\emptyset}(g)=1$. Then Proposition~\ref{proprationalabc} is proved by induction in~$|\bl|$. Indeed, using the induction assumption, one can show that each term in the right-hand side of~\eqref{Phirec} multiplying the factor $(2g-2|\bl|-2d-2)_{2|\bl|+2d+n+1} \, \prod_{i=1}^{|\bl|+d}(2g-2i+1)^2$ is a polynomial in $\mathbb{Z}[1/2][g]$ . This finishes the proof.
\end{proof}

\begin{prop}\label{proprational}
For $n\ge1$ and for a given fixed partition $\bl=(\lambda_1,\dots,\lambda_{n-1})$, we have
\begin{align}\label{eqrational2}
\frac{C(\bl,g-1-|\bl|)}{C(g-1)} \= \frac{Q_{\bl}(g)}{\prod_{j=1}^{n-3}(2g+j)} \, \prod_{i=1}^{|\bl|}\frac{2g-2i}{(2g-2i+1)^3}\,,
\end{align}
where $Q_{\bl}(g)\in \mathbb{Z}[1/2][g]$ are polynomials. 
\end{prop}
\begin{proof}
For $n=1$ we have $Q_{\emptyset}(g)\= 1$.
Assume that $n\ge3$ and $n$ is odd.
Write $\dd=(\bl,g-1-|\bl|)$.  From Proposition~\ref{prop2} we know that 
\begin{align}\label{npointformula}
\frac{(2g+n-3)!}{2^{2g-1}}\, C(\mathbf{d}) \= \sum_{\substack{\sigma\in S_n \\ \sigma(n)=n}} (-1)^{|S_{\sigma}^{+}|+1} \sum_{\substack{k_1,\dots,k_{n}\ge-1\\ k_1+\cdots+k_{n}=d_1+\cdots+d_{n}}} a_{k_1,\dots,k_{n}}
\,\omega_{\mathbf{d},\sigma,\mathbf{k}}\,,
\end{align}
where $a_{k_1,\dots,k_n}$ is defined in~\eqref{defak1kn}, and $\omega_{\dd,\sigma,\mathbf{k}}$ is defined in~\eqref{omegadsigmak}.
By the definition~\eqref{defak1kn} of $a_{k_1,\dots,k_n}$, we know that
\begin{align}\label{ak1kninZ1/2}
\frac{8^{k_1+\cdots+k_{n}}\,(k_{n}+1)!\,a_{k_1,\dots,k_{n}}}{(2k_{n}-1)!!^3}\in \mathbb{Z}[1/2][k_{n}]\,.
\end{align}
Now we claim that all nonzero summands in the RHS of~\eqref{npointformula} correspond to \beq g-|\bl|\, \leq \,k_n \,\leq\, g+\frac{n-3}2\,. \label{knbound}\eeq Indeed, one can show that $k_n<g-|\bl|$ implies  $\omega_{\underline{d},\sigma,\underline{k}}=0$, and that $k_n>g+\frac{n-3}2$ implies $a_{k_1,\dots,k_n}=0$.
Therefore, we obtain from~\eqref{ak1kninZ1/2} and~\eqref{knbound} that
\begin{align}\label{polynomial0723}
\frac{8^{g-1}\,(g+\frac{n-1}2)!}{(2g-2|\lambda|-1)!!^3}\,\sum_{\substack{\sigma\in S_n \\ \sigma(n)=n}} (-1)^{|S_{\sigma}^{+}|+1} \sum_{\substack{k_1,\dots,k_{n}\ge-1\\ k_1+\cdots+k_{n}=d_1+\cdots+d_{n}}} a_{k_1,\dots,k_{n}}
\,\omega_{\mathbf{d},\sigma,\mathbf{k}}\,\in\, \mathbb{Z}[1/2][g]\,.
\end{align}
Moreover, we notice that the polynomial~\eqref{polynomial0723} has zeros at $g=|\lambda|,|\lambda|-1,\dots,-\frac{n-1}2$, (since at these points, $g-1-|\bl|$ is a negative integer and the polynomials~\eqref{polynomial0723} equals~0 by the discussion in Section~\ref{secpartitionfixedlength}), so we obtain that
\begin{align}\label{Cpoly}
\frac{(2g+n-3)!}{2^{2g-1}}\frac{8^{g-1}\,(g-|\lambda|+1)!}{(2g-2|\lambda|-1)!!^3}\,C(\mathbf{d})\,\in\, \mathbb{Z}[1/2][g]\,.
\end{align}
Using~\eqref{Cpoly} and the one-point formula~\eqref{smallest}, we deduce~\eqref{eqrational2} for $n\ge3$ odd. For $n$ even and $n\ge4$, the proof is similar and for $n=2$, the proof is based on the fact that $C(d_1,d_2)=C(0,d_1,d_2)$. 
\end{proof}
We note that the above two propositions are analogous to results of Liu--Xu \cite{LX} on 
Witten's intersection numbers, and that 
in~\cite{GY} we used the matrix-resolvent formula to prove rationality for Witten's intersection numbers.

We make one further remark.  In the formulas given before Proposition~\ref{proprationalabc} 
we see that the polynomials $Q_{\bl}(g)$ occurring for different~$\bl$ of the same length~$m$
are very close for~$g$ large, e.g., the three polynomials $Q_{\bl}(g)$ 
for $m=3$ all start $1024g^7-7168g^6$, and even their $g^5$ coefficients are near each other.
We will return to this point in more detail later.

It will be convenient to write the polynomials $Q_{\bl}(g)$ as polynomials of the variable~$X$, i.e., 
$Q_{\bl}(g)=P_{\bl}(X)$, $X=2g-2+n$. For example, 
\begin{align}
&P_{1}(X) \= X^2-X+\frac32\,, \nn\\
&P_{2}(X) \= X^4 - 6 X^3 + \frac{27}2 X^2 - \frac{27}2 X + \frac{135}8\,, \nn\\
&P_{1,1}(X) \= X^4 - 10 X^3 + 38 X^2 - 68 X + \frac{273}4\,, \nn\\
&P_3(X) \= X^6 - 15 X^5 + \frac{177}2 X^4 - 260 X^3 + \frac{3375}{8} X^2 - \frac{3375}{8} X + \frac{7875}{16} \,, \nn\\
&P_{1,2}(X)\= X^6 - 21 X^5 + 179 X^4 - 795 X^3 + \frac{15957}{8} X^2 - \frac{23247}8 X + \frac{41121}{16}\,, \nn\\
&P_{1,1,1}(X)\= X^7 - 28 X^6 + \frac{653}{2} X^5 - \frac{4109}2 X^4 + \frac{30361}4 X^3 - \frac{33581}2 X^2 + \frac{170757}8 X - \frac{82467}8\,. \nn 
\end{align}
The statement of Proposition~\ref{proprational} can then be written equivalently as % the following proposition. 
\begin{align}\label{eqrational}
\frac{C(\bl,g-1-|\bl|)}{C(g-1)} \= \frac{P_{\bl}(2g+n-2)}{\prod_{j=1}^{n-3}(2g+j)} \, \prod_{i=1}^{|\bl|}\frac{2g-2i}{(2g-2i+1)^3} \,,
\end{align}
where $P_{\bl}(X)\in \mathbb{Z}[1/2][X]$ is a monic polynomial.

By using formula~\eqref{smallest} and formula~\eqref{eqrational} we arrive at the following proposition. 
\begin{prop} \label{proprationality}
For any fixed $n\ge2$, fixed $\lambda=(\lambda_1,\dots,\lambda_{n-1})\in(\mathbb{Z}_{\ge1})^{n-1}$, and for $d_n$ being an indeterminate, we have
	\begin{align}\label{UtoP}
		C(\lambda,d_n) \=  \frac{1}{X(\lambda,d_n)^{\delta_{n,2}} \,(X(\lambda,d_n)-1)! } \, \frac{(2d_n+1)!!^3}{2^{d_n+1} \, d_n!}  \, P_{\lambda}(X(\lambda,d_n))\,,
	\end{align}
	where $P_{\lambda}(X)\in \mathbb{Z}[1/2][X]$ is a monic polynomial of degree $\sum_{j=1}^{n-1}(2\lambda_j+1)+\delta_{n,2}-2$. 
\end{prop}

We observe from the above examples
that the polynomials $P_{\bl}(X)$ for $|\bl|=1,2$ are irreducible over~$\Q$, and we have checked that 
it is still true for $|\bl|=3, \dots, 9$. We expect that this irreducibility holds for all partitions~$\bl$, 
	and this is consistent with an observation that the BGW numbers often contain large prime factors. 
\begin{remark}
For $\dd' = (d_1,\dots,d_{n-1})\in (\ZZ_{\ge0})^{n-1}$ fixed and $d_n$ an indeterminate, write $\dd=(\dd',d_n)$. Then from Proposition~\ref{proprationality} and~\eqref{CC0} we have 
\begin{align}\label{Cdexpansionincluding0}
C(\mathbf{d}) \=  \frac{1}{X(\mathbf{d})^{\delta_{n,2}} \,(X(\mathbf{d})-1)! } \, \frac{(2d_n+1)!!^3}{2^{d_n+1} \, d_n!}  \, P_{\dd'}(X(\mathbf{d}))\,,
\end{align}
Here $P_{\dd'}(X)$ is defined 
via
\begin{align}
P_{0,\dd'}(X) \: (X-1)^{1-\delta_{n,2}} P_{\dd'}(X-1)\,,\quad  \dd'\in(\mathbb{Z}_{\ge0})^{n-1}\,.
\end{align} 
\end{remark}

Formula~\eqref{Cdexpansionincluding0} (or Proposition~\ref{proprationality})
 implies the following corollary,  which is similar to a result of Liu--Xu~\cite{LX2, LX} for Witten's intersection numbers. 
\begin{cor}\label{corintegrality}
For $g,n\ge1$, $d_1,\dots,d_n\ge0$ satisfying $d_1+\cdots+d_n=g-1$, we have 
\begin{align}\label{normaltobeinteger}
g^{\delta_{n,2}} \, \frac{d_n!\, \prod_{j=1}^n(2d_j+1)!!}{(2d_n+1)!!^3}  \, \langle\tau_{d_1}\cdots\tau_{d_n}\rangle_g  \in \mathbb{Z}[1/2]\,.
\end{align}
\end{cor}

In the introduction %(after~\eqref{BGWnumbers}) 
we formulated the ``Integrality Conjecture" that the
numbers $\langle\tau_{d_1}\cdots\tau_{d_n}\rangle_g$ are integral away from~2, and also made 
the observation  that they are often
highly factorized.  Corollary~\ref{corintegrality} does not imply the Integrality Conjecture, but
does give both some bounds on the denominators of $\langle\tau_{d_1}\cdots\tau_{d_n}\rangle_g$ and nice
information about prime factors of $\langle\tau_{d_1}\cdots\tau_{d_n}\rangle_g$.  The following statement, based
on the numerical data up to genus~40, gives a stronger version of the Integrality Conjecture. 

\begin{conj}\label{conjfactorsBGW1}
For $n\ge1$ and ${\bf d}=(d_1,\dots,d_n)\in (\ZZ_{\ge0})^n$, we have both 
\begin{align}\label{taudivisbiltiy1}
\langle\tau_{d_1}\cdots\tau_{d_n}\rangle_g^{\Theta} \;  \in\; \frac{\prod_{j=1}^n d_j!}{2^{4g}}\;\mathbb{Z} 
\end{align}
and
\begin{align}\label{taudivisbiltiy2}
\langle\tau_{d_1}\cdots\tau_{d_n}\rangle_g^{\Theta} \;  
\in\; \frac{\max_{1\leq j\leq n}\{(2d_j+1)!!\} \, \prod_{r\ge0, \, p_r(\dd)\ge1} (p_r(\dd)-1)!}{2^{4g}} 
\;\mathbb{Z} \,,
\end{align}
where  $g=|{\bf d}|+1$, and $p_r(\dd)$ denotes the multiplicity of~$r$ in~$\dd$.
\end{conj}
We notice that none of \eqref{normaltobeinteger}, \eqref{taudivisbiltiy1}, \eqref{taudivisbiltiy2} imply either of the others, 
and also that each is weaker than the best possible factorization. For instance,
\begin{align}
\langle\tau_6\tau_7\tau_8\tau_{18}\rangle_{40}^{\Theta} \= 2^{-150}\thin 3^{18}\thin 5^{11}\thin 7^5\thin 11^2\thin 13^3\thin 17^1\thin 19^2\thin 23^2\thin29^2\thin31^2\thin37^2\,101^1\, M\,, \nn
\end{align}
where $M$ has no prime factors less than~$1000$,
whereas formula~\eqref{normaltobeinteger} implies the
divisibility of  $\langle\tau_6\tau_7\tau_8\tau_{18}\rangle_{40}^{\Theta}$ (away from~2) by $3^{-1}\thin 13^{-2}\thin5^{2}\thin7^1\thin 19^2\thin 23^2\thin 29^2\thin 31^2\thin 37^2$, 
formula~\eqref{taudivisbiltiy1} implies the divisibility by
$3^{14}\thin 5^6\thin 7^4\thin 11^1\thin 13^1\thin 17^1$, and formula~\eqref{taudivisbiltiy2} 
implies the divisibility by $3^{9}\thin 5^5\thin 7^3\thin 11^2\thin 13^1\thin 17^1\thin 19^1\thin 23^1\thin 29^1\thin 31^1\thin 37^1$.

\section{Uniform large genus asymptotics} \label{secproofofuniformleadingasymp}
In this section, we prove Theorem~\ref{thmuniformasymptotics}. 
Our proof will mainly use the recursion~\eqref{DNrec}, and
techniques introduced by Aggarwal~\cite{Agg} in the study of the large genus 
asymptotics of Witten's intersection numbers. 
Before entering into the details, it is convenient to rewrite the DVV-type relation~\eqref{DNrec} in terms of~$C(\mathbf{d})$ as follows: 
\begin{align}
	&C({\bf d}) \=\sum_{j=2}^n\frac{2d_j+1}{X(\mathbf{d})-1}C(d_2,\dots,d_j+d_1,\dots,d_n) \nn\\
	& +\sum_{\substack{a,b\ge0\\a+b=d_1-1}}\biggl[\frac{2}{X(\mathbf{d})-1}\,C(a,b,d_2,\dots,d_n) \nn\\
	& +\sum_{I\sqcup J=\{2,\dots,n\}}\frac{(X(a,\mathbf{d}_{I})-1)!\,(X(b,\mathbf{d}_J)-1)!}{(X(\mathbf{d})-1)!} \, C(a,\dd_I)\,C(b,\dd_J)\biggr]\,,  \label{CVirasoro2}
\end{align}
where $n\ge1$, $\mathbf{d}=(d_1,\dots,d_n)\in(\mathbb{Z}_{\ge0})^n$, and $X(\cdot)$ is as in~\eqref{defXd}. 

\subsection{Lower bound}
Let us first show the following lemma on positivity of $C(\dd)$.
\begin{lemma}\label{positivitylemma1020}
For every $n\ge1$ and every $\mathbf{d}=(d_1,\dots,d_n)\in\left(\mathbb{Z}_{\ge0}\right)^n$, we have
	\begin{align}\label{positivity1020}
		C(\dd) \,>\, 0 \,.
	\end{align}
\end{lemma}
\begin{proof}
By using the DVV-type relation~\eqref{CVirasoro2} and by recalling that $C(0)=1/4$. 
\end{proof}
Now we give in the following lemma a better lower bound for $C(\dd)$.
\begin{lemma}\label{lemmalowerbound}
	For every $n\ge1$ and every $\mathbf{d}=(d_1,\dots,d_n)\in\left(\mathbb{Z}_{\ge0}\right)^n$, we have
	\begin{align}\label{lowerbound}
		C(\dd)\,\geq\, C(|\dd|) \,.
	\end{align}
\end{lemma}
\begin{proof}
Noticing that the statement is trivial when $n=1$ and the fact that  
 $C(\dd)$ is unchanged by removing~0's in~$\dd$, we can assume $n\ge2$ and $d_1,\dots,d_n\ge1$.
 Because of the symmetry of $C(\dd)$ in its arguments, we can also assume that $d_1$ is the smallest of the $d_j$'s.
We now do mathematical induction with respect to $X(\mathbf{d})=2|\dd|+n\ge3$.  
For $X(\dd)=3$, from Table~\ref{tablenormalizednumbers} we see that \eqref{lowerbound} is true.
By Lemma~\ref{positivitylemma1020} 
and by the induction hypothesis
we get
\begin{align}\label{Cd1dnlowerbound}
C(\dd)- C(|\dd|) \;\ge\; - 2 \, d_1 \, \frac{C(|\dd|)-  C(|\dd|-1)}{X(\mathbf{d})-1} \+ 2 \, \frac{C(0) \, C(|\dd|-1)}{(X(\mathbf{d})-1)(X(\mathbf{d})-2)} \,.
\end{align}
Using formula~\eqref{smallest}, we can write the right-hand side of~\eqref{Cd1dnlowerbound} as
\begin{align}\label{lowerboundest1}
\frac12 \, \biggl(\frac{1}{X(\mathbf{d})-2}-\frac{d_1}{|\dd|(|\dd|+1)}\biggr) \, \frac{C(|\dd|-1)}{X(\mathbf{d})-1}\,,
\end{align}
which is positive because $d_1\leq \frac{|d|}n$ and $n\ge2$. This finishes the proof.
\end{proof}

\subsection{Upper bound}
To give an upper bound of $C(\dd)$, we first give in the following lemma an estimate 
related to the quadratic-in-$C$ terms of~\eqref{CVirasoro2}, analogous to~\cite[Lemma~3.1]{Agg}.
\begin{lemma}\label{estimatethird}
	For $n\ge1$ and for $\mathbf{d}=(d_1,\dots,d_n)\in \bigl(\mathbb{Z}_{\ge1}\bigr)^n$, we have
\begin{align}\label{upperboundthirdterm}
\sum_{\substack{a,b\ge0\\a+b=d_1-1}}\sum_{I\sqcup J=\{2,\dots,n\}}\frac{\bigl(X(a,\mathbf{d}_{I})-1\bigr)!\,\bigl(X(b,\mathbf{d}_{J})-1\bigr)!}{\bigl(X(\mathbf{d})-1\bigr)!}\;\leq\; \frac{4}{(X(\mathbf{d})-1)(X(\mathbf{d})-2)}\,.
	\end{align}
\end{lemma}
\begin{proof}
For $n=1$, the inequality~\eqref{upperboundthirdterm} is trivial. Assume that $n\ge2$. For each $a,b,I,J$ satisfying $a+b=d_1-1$ and $I\sqcup J=\{2,\dots,n\}$, we set $n_1=|I|$, $g_1=a+1+|\mathbf{d}_{I}|$, $n_2=|J|$, $g_2=b+1+|\mathbf{d}_{J}|$. Then $X(a,{\bf d_I})=2g_1+n_1-1$ and $X(b,{\bf d_J})=2g_2+n_2-1$. %Using these variables and 
By counting the number of 4-tuples $(a,b,I,J)$ with given values of $n_i$ and $g_i$, we find
\begin{align}
&\sum_{\substack{a,b\ge0\\a+b=d_1-1}}\sum_{I\sqcup J=\{2,\dots,n\}}\frac{\bigl(X(a,\mathbf{d}_{I})-1\bigr)!\,\bigl(X(b,\mathbf{d}_{J})-1\bigr)!}{\bigl(X(\mathbf{d})-1\bigr)!} \nn\\
\leq &\sum_{n_1+n_2=n-1}\sum_{\substack{g_1,g_2\ge1\\g_1+g_2=g }}\binom{n-1}{n_1}\frac{(2g_1+n_1-2)!\,(2g_1+n_2-2)!}{(X(\mathbf{d})-1)!}\nn\\
= &\sum_{n_1+n_2=n-1}\binom{n-1}{n_1}\biggl(2\,\frac{n_1!\,(X(\dd)-n_1-3)!}{(X(\mathbf{d})-1)!} \,+\,  \sum_{\substack{g_1,g_2\ge2\\g_1+g_2=g }}\frac{(2g_1+n_1-2)!\,(2g_1+n_2-2)!}{(X(\mathbf{d})-1)!}\biggr)\,,
		\label{estimatethirdterm}
	\end{align}
where $g=g(\dd)=|d|+1$ as usual. We estimate the two terms on the right-hand side of~\eqref{estimatethirdterm} separately. For the first term we have 
\begin{align}\label{estimatethirdterm2}
&2\sum_{n_1+n_2=n-1}\binom{n-1}{n_1}\frac{n_1!\,(X(\mathbf{d})-n_1-3)!}{(X(\mathbf{d})-1)!}
\=\frac{2}{(X(\mathbf{d})-1)(X(\mathbf{d})-2)}\sum_{n_1=0}^{n-1}\prod_{j=1}^{n_1}\frac{n-j}{X(\mathbf{d})-2-j}\nn \\
&\qquad\qquad \;\leq\; \frac{2}{(X(\mathbf{d})-1)(X(\mathbf{d})-2)}\sum_{n_1=0}^{\infty}\left(\frac{1}{3}\right)^{n_1}\=\frac{3}{(X(\mathbf{d})-1)\,(X(\mathbf{d})-2)}\,, 
\end{align}
where in the inequality we used the fact that $n\leq \frac{X(\mathbf{d})}{3}$ (implied by  $\dd\in(\mathbb{Z}_{\ge1})^n$).
For the second term we have 
\begin{align}
&\sum_{n_1+n_2=n-1}\sum_{\substack{g_1,g_2\ge2\\g_1+g_2=g }}\binom{n-1}{n_1}\frac{(2g_1+n_1-2)!\,(2g_2+n_2-2)!}{(X(\mathbf{d})-1)!}\nn\\
&\leq \sum_{\substack{g_1,g_2\ge2\\g_1+g_2=g \\n_1+n_2=n-1 }} \frac{\binom{2g-4}{2g_1-2}^{-1}}{(X(\mathbf{d})-1)(X(\mathbf{d})-2)}\leq  \frac{n\, (g-3) \,\binom{2g-4}{2}^{-1}}{(X(\mathbf{d})-1)(X(\mathbf{d})-2)} \leq \frac{1}{(X(\mathbf{d})-1)\,(X(\mathbf{d})-2)}\,,\label{estimatethirdterm1}
\end{align}
where for the first inequality we used the fact that
$X(\dd)=(2g_1+n_1-2)+(2g_2+n_2-2)$ and that
\begin{align}
\binom{a_1}{b_1}\binom{a_2}{b_2}\leq \binom{a_1+a_2}{b_1+b_2}\,,
\end{align}
and for the last inequality we used $g\ge3$ and $n\leq g-1$. % (implied by $\dd\in(\mathbb{Z}_{\ge1})^n$).  
Combining \eqref{estimatethirdterm},~\eqref{estimatethirdterm2},~\eqref{estimatethirdterm1}, we obtain the lemma.
\end{proof}

Following Aggarwal~\cite{Agg}, for $\X,n\ge1$, introduce 
\begin{align}\label{defthetagn}
	\theta_{\X,n} \; := \;
	\max_{\substack{\dd\in(\mathbb{Z}_{\ge0})^n \\ X(\mathbf{d})=\X}}C(\dd) \,.
\end{align} 
Before continuing, we introduce a number-theoretic function $f(\X,n)$, defined through the recursion
\begin{align}\label{deff*}
f(\X,n)\=
\frac{2}{3}f(\X-1,n-1)+\frac{1}{3}f(\X-1,n+1)+\frac{4}{(\X-1)(\X-2)}\,, \quad \forall\, n\ge3,\X\ge8\,,
\end{align} 
together with the initial data $f(\X,n)=1/\pi$
for $1\leq \X\leq 7$ or $n=1,2$. By induction, we know that $f(\X,n)$ is monotone increasing with respect to~$n$, and 
\begin{align}
f(\X,n) \,\leq\, \frac{1}{\pi} \+\sum_{k=8}^{\X} \frac{4}{(k-1)(k-2)}\,.
\end{align}
Hence $f(\X,n)$ is bounded by~1. Let us also prove the following lemma.
\begin{lemma}\label{lemsupnthetag,n}
For $\X\ge1$, and for $1\leq n\leq \frac{\X}5$, we have the uniform estimate
\begin{align}
f(\X,n) \= \frac{1}{\pi}\+O\bigl(\frac{1}{\X}\bigr)\,,\quad \X\to\infty \,.
\end{align}
\end{lemma}
\begin{proof}
It is not difficult to show, either directly or using generating function, that the solution of the recursion~\eqref{deff*} with the given initial condition, is given for $n\ge2$ and $\X\ge7$ by
	\begin{align}\label{expressionf*1}
		f(\X,n) \= \frac{1}{\pi} +\frac{2(\X-7)}{3(\X-1)} -\sum_{k=8}^{\X} \frac{2(k-7)}{3(k-1)} P(n,\X-k)\,,
	\end{align}
	where the coefficients $P(n,j)$ are defined by
	\begin{align}
		\Bigl(\frac{3-\sqrt{9-8t^2}}{2t}\Bigr)^{n-2} \;=:\;
		\sum_{j=0}^{\infty} P(n,j) \, t^j\,.
	\end{align}
These coefficients can be estimated by the residue theorem:
		\begin{align}
			&P(n,j)= \frac{1}{2\pi i}
			\int_{C}t^{-j-1}\biggl(\frac{3-\sqrt{9-8t^2}}{2t}\biggr)^{n-2} \mathrm{d} t\;\leq\;  1.05^{-j} \times 1.23^{n-2}\,,\nn
		\end{align}
where we have taken the contour $C$ to be the circle $|t|=1.05<\sqrt{9/8}$, on which
we have $\biggl|\frac{3-\sqrt{9-8t^2}}{2t}\biggr|<1.23$.
Let us now estimate the right-hand side of~\eqref{expressionf*1}:
\begin{align}\label{Pestimate1}
		\sum_{k=8}^{\X} \frac{2(k-7)}{3(k-1)}P(n,\X-k)
		\;\geq\; \sum_{k=M+1}^{\X} \frac{2(k-7)}{3(k-1)}P(n,\X-k)
		\;\geq\; \frac{2(M-6)}{3M}\sum_{j=0}^{\X-M-1}P(n,j)\,,
\end{align}
for any $7\leq M\leq \X$.	Taking $M=[0.1\X]$ and using $n\leq \frac{\X}5$, we have
\begin{align}\label{Pestimate2}
		\sum_{j=\X-M}^{\infty} P(n,j) \leq 14\times 1.05^{-\X+M} \times 1.23^{n} \leq 14\times 0.998^\X\,.
\end{align}
By using $\sum_{j=0}^{\infty} P(n,j)=1$ and concluding formula~\eqref{expressionf*1} and the estimates~\eqref{Pestimate1},~\eqref{Pestimate2},  we finish the proof of the lemma.
\end{proof}	

The significance of the function $f(\X,n)$ is given by the following important lemma.
\begin{lemma}\label{lemineqthetaXn}
For $n\ge1$, $\X\ge1$, the numbers $\theta_{\X,n}$ have the upper bound
\begin{align}\label{ineqthetaXn}
\theta_{\X,n}\leq f(\X,n)\,.
\end{align}
\end{lemma}
\begin{proof}
For $1\leq \X\leq 7$, we check from Table~\ref{tablenormalizednumbers} that inequality~\eqref{ineqthetaXn} holds. For $n=1$, inequality~\eqref{ineqthetaXn} is implied by~\eqref{smallest}. For $n=2$, let $0\leq d_1\leq d_2$ and $g=d_1+d_2+1$. Using~\eqref{TwoPointsC} we get
\begin{align}
	&C(d_1,d_2) \= \frac{2^{2g}}{(2g)!} \, \sum_{h=0}^{d_1}(g-2h)\,F_h\,F_{g-h} \nn\\
	&\;\leq\; \frac{2^{2g}}{(2g)!} \,\Bigl(gF_g +\frac{(g-1)(g-2)}{2}F_1F_{g-1}\Bigr)\= \frac{(16 g^2-23g+9)\, \Gamma (g-\frac{1}{2})^2}{8\,\pi\, (2 g-1)\, \Gamma (g)^2}
	\;\leq\; \frac{1}{\pi}\,,
\end{align}
where in the first inequality we used that $F_{d+1}/F_d$ is monotone increasing for $d\ge0$.

Now consider the case that $n\ge3$, $\X\ge8$. Let us prove inequality~\eqref{ineqthetaXn} by induction on~$\X$. 
For every $\dd=(d_1,\dots,d_{n})\in \bigl(\mathbb{Z}_{\ge0}\bigr)^{n}$ satisfying $X(\dd)=\X$, 
assume $\dd\in \bigl(\mathbb{Z}_{\ge1}\bigr)^{n}$ and without loss of generality assume $d_1=\min \{d_j\}$. Then
by applying Lemma~\ref{estimatethird} to~\eqref{CVirasoro2} and by induction, we obtain
\begin{align}\label{thetag,nbound}
C(\dd)\;
\leq\; \Bigl(1-\frac{2d_1}{\X-1}\Bigr) \, f(\X-1,n-1)+\frac{2d_1}{\X-1} \, f(\X-1,n+1)+\frac{4}{(\X-1)(\X-2)}\,,
\end{align}
where we have used the fact that $f(\X,n)$ is bounded by~1.
This gives
\begin{align}
C(\dd)\;\leq\; \frac{2}3 \, f(\X-1,n-1)+\frac{1}{3} \, f(\X-1,n+1)+\frac{4}{(\X-1)(\X-2)}\= f(\X,n)\,,\nn
\end{align}
where we have used $n\ge3$, $d_1\leq \frac{|\dd|}n= \frac{\X-n}{2n}$, and the fact that the function~$f(\X,n)$ is increasing with respect to~$n$.
For the case that some of $d_j$ equal zero, by~\eqref{CC0} and by induction we have
\begin{align}
C(\dd) \;\leq\; f(\X-1,n-1) \,,
\end{align}
where the right-hand side is less than $f(\X,n)$ because of~\eqref{deff*} and again the monotonicity of $f(\X,n)$.
Combining both cases, we finish the  proof of~\eqref{ineqthetaXn}.
\end{proof}

We are ready to prove Theorem~\ref{thmuniformasymptotics}.

\begin{proof}[Proof of Theorem~\ref{thmuniformasymptotics}]
By Lemma~\ref{lemmalowerbound} and~\eqref{defthetagn}, we have the lower and upper bound
\begin{align}\label{lowerupperbound}
C(|\dd|)\leq C(\dd) \leq \theta_{X(\dd),n}\,,
\end{align}
for every $n\ge1$, $\dd\in(\mathbb{Z}_{\ge0})^n$. We know from~\eqref{smallest} that the lower bound equals~$\frac{1}{\pi}+O(\frac{1}{g(\dd)})$ with an absolute constant in $O(\frac{1}{g(\dd)})$. For the upper bound, we know from Lemma~\ref{lemsupnthetag,n} 
that when $n\leq \frac{X(\dd)}{5}$, $\theta_{\X,n}$ is bounded by $\frac{1}{\pi}+O(\frac{1}{\X})$ with an absolute constant in~$O(\frac{1}{\X})$. 
Consider the case that $\dd\in(\mathbb{Z}_{\ge0})^n$ with $\frac{X(\dd)}5<n\leq \frac{X(\dd)}3$. Assume $\dd\in(\mathbb{Z}_{\ge1})^n$. Then there must exist some $j$ such that $d_j=1$. Without loss of generality, assuming that $d_1$ equals~1, then recursion~\eqref{CVirasoro2} reads
\begin{align}
&C(1,d_2,\dots,d_n) \=  \sum_{j=2}^n\, \frac{2d_j+1}{X(\dd)-1} \, C(d_2,\dots,d_j+1,\dots,d_n) \nn\\
&  \+\frac{2}{X(\dd)-1} \, C(d_2,\dots,d_n) \+\sum_{I\sqcup J=\{2,\dots,n\}}
\frac{\bigl(X(0,\mathbf{d}_{I})-1\bigr)!\,\bigl(X(0,\mathbf{d}_{J})-1\bigr)!}{(X(\dd)-1)!} \, C(\dd_I)\,C(\dd_J). \label{C1d}
\end{align}
By applying Lemma~\ref{estimatethird} and by taking maximum in both sides of~\eqref{C1d}, we get
\begin{align}
C(\dd) \,\leq \,  \theta_{X(\dd)-1,n-1} \+ \frac{4}{(X(\dd)-1)(X(\dd)-2)}\,.\label{Cddupperbound}
\end{align}
When some $d_j$ is~0, the inequality~\eqref{Cddupperbound} is still true according to~\eqref{CC0}, so we get
\begin{align}\label{thetaX,nest}
\theta_{\X,n} \,\leq\,  \theta_{\X-1,n-1} \+ \frac{4}{(\X-1)(\X-2)}\,,
\end{align}
for every $\frac{\X}{5}< n\leq \frac{\X}{3}$.
Writing this as 
\begin{align}
\theta_{\X,n}+\frac{4}{\X-1} \,\leq\, \theta_{\X-1,n-1}+\frac{4}{\X-2}\,,
\end{align}
and iterating~$t$ times we find 
\begin{align}
\theta_{\X,n} \leq \theta_{\X-t,n-t} - \frac{4}{\X-1} \+ \frac{4}{\X-1-t} \leq f(\X-t,n-t) - \frac{4}{\X-1} \+ \frac{4}{\X-1-t}
\end{align}
for any $t\leq [\frac{5n-\X+1}4]$. Applying this with $t=[\frac{5n-\X+1}4]$ and
using Lemma~\ref{lemsupnthetag,n} we obtain that $\theta_{\X,n}$ is bounded by $\frac{1}{\pi}+O(\frac{1}{\X})$ uniformly when $\frac{\X}{5}< n\leq \frac{\X}{3}$. This together with~\eqref{lowerupperbound} implies that 
formula~\eqref{STRONG} holds for $\frac{X(\dd)}{5}\leq  n\leq \frac{X(\dd)}{3}$. 
For $n>\frac{X(\dd)}{3}$, we have $\theta_{X(\dd),n}=\theta_{X(\dd)-1,n-1}$ 
(indeed, since $C(\dd)$ is unchanged by removing any~0 argument, 
we have by definition $\theta_{\X,n}=\theta_{\X-1,n-1}$ for $n > \frac{\X}3$), which implies that formula~\eqref{STRONG} still holds. Combining all three cases, we obtain the statement of Theorem~\ref{thmuniformasymptotics}.
\end{proof}

\begin{remark}
Although our proof for BGW numbers 
 is similar to Aggarwal's for 
 Witten's intersection numbers~\cite{Agg}, 
 we have made several improvements and simplifications.
 For example, the technique of random walks used in~\cite{Agg} is avoid here, and our estimates are completely uniform instead of 
 requiring $n=o(\sqrt{g})$ in~\cite{Agg}. Actually, as it was shown in~\cite{DGZZ20}, it 
 is not possible to extend the asymptotics of the normalized Witten's intersection 
 numbers \cite{Agg,DGZZ20,DGZZ22,GY} beyond the range $n=o(\sqrt{g})$. 
 We hope to generalize Theorem~\ref{thmuniformasymptotics} to
 Witten's intersection numbers with a better normalization.
\end{remark}

\section{Polynomiality in large genus} \label{secproofofpolynomiality} 
In this section, we prove Theorem~\ref{thmpoly} by using the recursion~\eqref{CVirasoro2}. 
\begin{proof}[Proof of Theorem~\ref{thmpoly}]
We first allow the fixed $\dd'=(d_1,\dots,d_{n-1})\in(\mathbb{Z}_{\ge0})^{n-1}$. 
 Write $\dd=(\dd',d_n)$ with $d_n\ge0$.
By using~\eqref{Cdexpansionincluding0} and Stirling's formula we have
\begin{align}\label{Cdexpansion2}
	C(\dd) \= \frac1\pi \,\sum_{k=0}^{\infty} \frac{C_k(\dd')}{X(\dd)^k}\,, \qquad \text{as } d_n \to \infty\,,
\end{align}
where $C_k$ are functions of~$\dd'$.
By using the recursion~\eqref{CVirasoro2} and by performing Laurent expansions, we obtain 
\begin{align}
&C_k(d,\dd')-C_k(\dd') \= -\sum_{l=1}^{k-1}(-1)^{k-l}\binom{k-1}{l-1} \, C_{l}(d,\dd') \nn\\
& + \sum_{j=1}^{n-1}(2d_j+1)\, \bigl(C_{k-1}(d_1,\dots,d_j+d,\dots,d_{n-1})-C_{k-1}(\dd')\bigr) \nn\\
& + \sum_{\substack{a,b\ge0\\ a+b=d-1}}
\Biggl[2 \, \big(C_{k-1}(a,b,\dd')-C_{k-1}(\dd')\bigr) + \sum_{I\sqcup J=\{1,\dots,n-1\}} \sum_{l=0}^{k-2} \, \mathfrak{a}_{(a,\dd'_I),k,l} \, 
C_l(b,\dd'_{J})  \Biggr]\,,  \label{Chatkrecursion}
\end{align}
where $d\ge0$ 
and $\mathfrak{a}_{\mathbf{w},k,l}$ are numbers defined by
\begin{align}
&\mathfrak{a}_{\mathbf{w},k,l} \: \sum_{u=l}^{k-2} \biggl(4^{|\mathbf{w}|+1} \,
\binom{u-1}{l-1} \,X(\mathbf{w})^{u-l} \,S(k-u-1,X(\mathbf{w}))\, B\bigl(\mathbf{w}\bigr)\biggr)\,.\nn
\end{align}
Here $B(\mathbf{w})$ is defined in~\eqref{tauU} and $S(n,k)$ denotes the Stirling number of the second kind. Write
\begin{align}\label{Ck=ck}
C_k(\dd') \;=:\; \tilde{c}_k(p_0(\dd'),p_1(\dd'),\dots)\,,
\end{align}
where $p_r(\dd')$ denotes the multiplicity of~$r$ in~$\dd'$.
This defines functions $\tilde c_k(\mathbf{p})$, $k\ge0$, where ${\bf p}=(p_0,p_1,p_2,\dots)$.
Then formula~\eqref{Chatkrecursion} becomes
\begin{align}
& \tilde c_k(\mathbf{p}+\mathbf{e}_{d}) - \tilde c_k(\mathbf{p})= -\sum_{l=1}^{k-1}\left(-1\right)^{k-l}\binom{k-1}{l-1} \, \tilde c_{l}(\mathbf{p}+\mathbf{e}_{d}) \nn\\
& +\sum_{i\ge0}(2i+1)\,p_i \, \bigl(\tilde c_{k-1}(\mathbf{p}-\mathbf{e}_i+\mathbf{e}_{i+d})-\tilde c_{k-1}(\mathbf{p})\bigr) \nn\\
&  +\sum_{\substack{a,b\ge0\\ a+b=d-1}}
\Biggl[2 \, \sum_{l=0}^{k-1} \big(\tilde c_{k-1}(\mathbf{p}+\mathbf{e}_a+\mathbf{e}_b)-\tilde c_{k-1}(\mathbf{p})\bigr)\nn\\
&\quad + \sum_{\substack{\mathcal{E}(\mathbf{t}+\mathbf{e}_a)\leq k-1 \\ 0\leq t_r\leq p_r, \,r\ge0}} \sum_{l=0}^{k-2} \, \biggl(
\tilde c_l(\mathbf{p}-\mathbf{t}+\mathbf{e}_b) \, \alpha_{\mathbf{t}+\mathbf{e}_a,k,l} \,\prod_{i\ge0}\binom{p_i}{t_i}\biggr)\Biggr]\,, \quad d\ge0\,,\label{hatCkrecursion}
\end{align}
where $\mathcal{E}({\bf t})=\sum_{j=0}^\infty (2j+1)t_j$, $\alpha_{\mathbf{t},k,l}=\mathfrak{a}_{(0^{t_0}1^{t_1}2^{t_2}\cdots),k,l}$, and $\mathbf{e}_d$ denotes $(0,\dots,0,1,0,0,\dots)$ with ``1" appearing in the $(d+1)$th place. 

Let us now prove by induction that $\tilde c_k(\mathbf{p})$, $k\ge0$, belong to $\mathbb{Q}[p_0,p_1,\dots]$ and satisfy the degree estimates
\begin{equation}\label{degreeest1}
\deg \, \tilde c_k({\bf p}) \,\leq\, k-1\,,\quad k\ge1\,,
\end{equation} 
under the degree assignments $\deg p_d=2d+1$, $d\ge0$. For $k=0$, by using Theorem~\ref{thmuniformasymptotics} 
we know that $\tilde c_0(\mathbf{p})\equiv1$. Assume that for $1\leq l\leq k-1$, $\tilde c_l({\bf p})\in\mathbb{Q}\left[p_0,p_1,\dots\right]$ are polynomials satisfying 
$\deg \, \tilde c_l({\bf p}) \leq l-1$. 
Then for $k$ and for every $d\ge0$, the RHS of equation~\eqref{hatCkrecursion} are polynomials in~$p_0,\dots,p_{[(k-3)/2]}$.
Moreover, by the inductive assumption these polynomials are independent of~$d$ for every $d\ge k+1$, i.e.,
\begin{align}\label{diffck}
\tilde c_k(\mathbf{p}+\mathbf{e}_d) - \tilde c_k(\mathbf{p}) \= \left\{\begin{array}{ll}
f_{d}(p_0,\dots,p_{[(k-3)/2]})\,, \quad &d\le k\,,\\
g(p_0,\dots,p_{[(k-3)/2]})\,, \quad &d\ge k+1
\end{array}\right.
\end{align}
for some $f_d$ ($d\leq k$), $g$ belonging to $\mathbb{Q}[p_0,\dots,p_{[(k-3)/2]}]$. The compatibility of~\eqref{diffck} implies that $g(p_0,\dots,p_{[(k-3)/2]})\equiv A$ is a constant.  Solving~\eqref{diffck} we obtain that
$\tilde c_k$ have the form
\begin{align}\label{ckform1}
\tilde c_k(\mathbf{p}) \= h(p_0,\dots,p_{k}) \+ A \, n'(\mathbf{p})\,,
\end{align} 
where $h\in \mathbb{Q}[p_0,\dots,p_k]$ and $n'(\mathbf{p}):=\sum_{i\ge0}p_i$.
Now we aim to show that $A=0$. Consider equation~\eqref{hatCkrecursion} with $k$ replaced by $k+1$. Using a similar analysis and using~\eqref{ckform1}, we obtain that for every $d\ge 2k+2$,
\begin{align}\label{ck+1diff}
\tilde c_{k+1}(\mathbf{p}+\mathbf{e}_d) - \tilde c_{k+1}(\mathbf{p}) \= 4\, A \, d \+ A'\, n'(\mathbf{p}) \+ s(p_0,\dots,p_k)\,,
\end{align} 
where $A'\in \mathbb{Q}$ is a constant, and $s(p_0,\dots,p_k)$ is some polynomial in $\mathbb{Q}[p_0,\dots,p_k]$. 
This contradicts~\eqref{TwoPointsC} unless $A=0$. 
Therefore, 
\begin{align}
\tilde c_k(\mathbf{p})\in\mathbb{Q}[p_0,\dots,p_{k}].
\end{align}
Then taking $d\ge k+1$ in equation~\eqref{hatCkrecursion} gives
\begin{align}
&0=-\sum_{l=1}^{k-1}(-1)^{k-l}\binom{k-1}{l-1} \, \tilde c_{l}(\mathbf{p})  
\+\sum_{i\ge0}(2i+1)\,p_i \, \bigl(\tilde c_{k-1}(\mathbf{p}-\mathbf{e}_i)-\tilde c_{k-1}(\mathbf{p})\bigr) \nn\\
& +  \sum_{a=0}^{[\frac{k-3}2]}\Biggl(
4  \, \bigl(\tilde c_{k-1}(\mathbf{p}+\mathbf{e}_a)-\tilde c_{k-1}(\mathbf{p})\bigr) 
+ \sum_{\substack{\mathcal{E}(\mathbf{t}+\mathbf{e}_a)\leq k-1 \\ 0\leq t_r\leq p_r, \, r\ge0}} \sum_{l=0}^{k-2} \, \bigg(
\tilde c_l(\mathbf{p}-\mathbf{t}) \, \alpha_{\mathbf{t}+\mathbf{e}_a,k,l} \,\prod_{i\ge0}\binom{p_i}{t_i}\bigg)\Biggr).\label{hatCkidentity}
\end{align}
Using~\eqref{hatCkidentity} and~\eqref{hatCkrecursion}, we obtain
\begin{align}
&\Delta_{p_d} \tilde c_{k}(\mathbf{p})\= \m\sum_{l=1}^{k-1}(-1)^{k-l}\binom{k-1}{l-1} \, \Delta_{p_d} \tilde c_{l}(\mathbf{p}) 
\+\sum_{i\ge0}(2i+1)\,p_i \, \Delta_{p_{i+d}} \tilde c_{k-1}(\mathbf{p}-\mathbf{e}_i) \nn\\
& + \, 2 \, \sum_{a=0}^{d-1}
 \Delta_{p_a}\Delta_{p_{d-1-a}} \tilde c_{k-1}(\mathbf{p})
\, - \, 4 \, \sum_{a=d}^{[(k-3)/2]}  \Delta_{p_a} \tilde c_{k-1}(\mathbf{p}) \nn\\
&+ \, \sum_{a=0}^{d-1}\sum_{\substack{\mathcal{E}(\mathbf{t}+\mathbf{e}_a)\leq k-1 \\ 0\leq t_r\leq p_r\,(r\ge0)}} \sum_{l=0}^{k-2} \biggl(
\, \Delta_{p_{d-1-a}} \tilde c_l(\mathbf{p}-\mathbf{t}) \, \alpha_{\mathbf{t}+\mathbf{e}_a,k,l} \,\prod_{i\ge0}\binom{p_i}{t_i}\biggr) \nn\\
& +\, \sum_{a=d}^{[(k-3)/2]} 
\sum_{\substack{\mathcal{E}(\mathbf{t}+\mathbf{e}_a)\leq k-1 \\ 0\leq t_r\leq p_r\,(r\ge0)}} 
\sum_{l=0}^{k-2} \, \biggl(\tilde c_l(\mathbf{p}-\mathbf{t}) \, \alpha_{\mathbf{t}+\mathbf{e}_a,k,l} \,\prod_{i\ge0}\binom{p_i}{t_i}\biggr)\,, \quad d\ge0 \,. \label{hatCk2}
\end{align}
We find that each term of the RHS of~\eqref{hatCk2} is of degree less than or equal to $k-2-2d$ for every $d\ge0$, which implies 
$\deg \tilde c_k\leq k-1$. In particular, $\tilde c_k$ is a polynomial that only depends on $p_0,\dots,p_{[k/2]-1}$. 

Now restrict to the case when $\dd'$ is a fixed partition. 
From~\eqref{Cdexpansion2},~\eqref{degreeest} we know that 
\begin{align}
&\widehat{C}(\dd) \,\sim\, \sum_{k=0}^{\infty}\frac{\widehat{C}_k(\dd')}{X(\dd)^k}\,, 
  \qquad d_n\to\infty\,,\label{Chatexpansion}
\end{align}
where $\widehat{C}_k$ are functions of~$\dd'$ with $\widehat{C}_0\equiv1$.
Define $\widehat{c}_k(p_1,p_2,\dots)\in \mathbb{Q}[p_1,p_2,\dots]$, $k\ge0$, via
\begin{align}
\gamma(X)\,\sum_{k=0}^{\infty} \frac{\widehat{c}_k(p_1,p_2,\dots)}{X^k} \= \sum_{k=0}^{\infty} \frac{\tilde c_k(0,p_1,p_2,\dots)}{X^k}\,,
\end{align}
where the left-hand side is understood as a power series in~$X^{-1}$. 
It then follows from $\deg \tilde c_k\leq k-1$ that 
$\deg \widehat{c}_k\leq k-1$.
Using~\eqref{Chatexpansion},~\eqref{Cdexpansion2},~\eqref{Ck=ck}, we know that
\begin{align}
\widehat{C}_k(\dd') \=\widehat{c}_k(p_1(\dd'),p_2(\dd'),\dots)\,,
\end{align}
for all~$\dd'$. The statement that $\widehat{c}_k(0,0,\dots)=0$ follows from the fact that $C(d)=\gamma(2d+1)$. This finishes the proof of Theorem~\ref{thmpoly}.
\end{proof}
\begin{remark}
Formula~\eqref{Chatkrecursion} is analogous to a formula given in~\cite[Corollary~3.6]{LX}.
\end{remark}

For $\dd=(d_1,\dots,d_n)\in (\mathbb{Z}_{\ge0})^n$, we extend the definition of $\widehat{C}(\dd)$ in~\eqref{normalizehatc61} by
\begin{align}\label{defChatextend}
\widehat{C}(\dd) \: \frac{C(\dd)}{\gamma(X(\dd)-p_0(\dd))}\,.
\end{align}
Then the following corollary easily follows from Theorem~\ref{thmpoly}.
\begin{cor}\label{corChatdexpansion}
For any fixed $n\ge1$ and fixed $\dd'=(d_1,\dots,d_{n-1})\in (\mathbb{Z}_{\ge0})^{n-1}$, we have
\begin{align}
\widehat{C}(\dd) \;\sim\; \sum_{k=0}^{\infty} \frac{\widehat{c}_k(p_1(\dd'),p_2(\dd'),\dots)}{(X(\dd)-p_0(\dd'))^k}\,, \qquad X(\dd)\to\infty\,,
\end{align} 
where $\dd=(\dd',d_n)$, and $\widehat{c}_k(p_1,p_2,\dots)$ are the same polynomials as those in Theorem~\ref{thmpoly}.
\end{cor}

We note that, since for fixed $n\ge1$ and fixed $d_1,\dots,d_{n-1}\ge0$, 
$C(\mathbf{d})/C(|\mathbf{d}|)$ is a rational function of~$X(\dd)$ whose asymptotic expansion is  convergent, 
the polynomials $\widehat{c}_k(p_1,p_2,\dots)$ , $k\ge0$, 
contain all information of BGW numbers. 
We provide in Table~\ref{tablepoly} 
	some explicit values of $\widehat{c}_k(p_1,p_2,\dots)$, $c_k(p_1,p_2,\dots)$.
\begin{table}[h!]
	\begin{center}
		\renewcommand\arraystretch{1.3}
		\tabcolsep=0.6cm
		\begin{tabular}{|c|c|c|}
			\hline
			$k$ &  $c_{k}(p_1,p_2,\dots)$ & $\widehat{c}_{k}(p_1,p_2,\dots)$ \\ 
			\hlinew{1.1pt}
			0 & 1 & 1 \\
			\hline
			1 & $-\frac{1}{2}$ & 0 \\ 
			\hline
			2 & $\frac{5}{8}$ & 0 \\
			\hline
			3 & $-\frac{11}{16}$ & 0 \\
			\hline
			4 & $\frac{83}{128}-\frac{27}{8}\,p_1$ & $-\frac{27}{8}\,p_1$ \\
			\hline
			5 & $-\frac{143}{256}-\frac{81}{16}\,p_1$ & $-\frac{27}{4}\,p_1$\\
			\hline
			6 & $\frac{625}{1024}-\frac{639}{64}\,p_1-\frac{1125}{16}\,p_2$ & $-\frac{45}4 \, p_1 - \frac{1125}{16} \, p_2$ \\
			\hline
			7 & $-\frac{1843}{2048}+\frac{25533}{128}p_1 -\frac{1701}{8}p_1^2-\frac{19125}{32}p_2$ & $\frac{783}4 \, p_1 - \frac{1701}8 \, p_1^2 - \frac{10125}{16} \, p_2$ \\
			\hline
			%				8 & $\frac{24323}{32768} +\frac{2912031}{1024}p_1-\frac{366687}{128} p_1^2 -\frac{111375}{32} p_2
			%				-\frac{385875}{128} p_3$ & $\frac{188559}{64} \, p_1 - \frac{380295}{128} \, p_1^2 - \frac{480375}{128} \, p_2 - \frac{385875}{128} p_3$\\
			%				\hline
		\end{tabular}
	\end{center}
	\caption{Expressions for $c_{k}(p_1,p_2,\dots)$ and $\widehat{c}_{k}(p_1,p_2,\dots)$ with $k\leq7$}
	\label{tablepoly}
\end{table}

Similar to~\eqref{Chatd1dnuniformindwithout0},  
for fixed $L\ge0$ and fixed $n\ge1$, and for $\dd=(d_1,\dots,d_n)\in(\ZZ_{\ge0})^n$,
\begin{align}\label{Chatd1dnuniformind}
	\widehat{C}(\dd) \= \sum_{k=0}^{L-1} \frac{\widehat{c}_{k}(p_1(\dd),p_2(\dd),\dots)}{(X(\mathbf{d})-p_0(\dd))^k} + O\biggl(\frac{1}{(X(\mathbf{d})-p_0(\dd))^{L}}\biggr)\,,  \qquad g(\mathbf{d})\to\infty\,,
\end{align}
where the implied O-constant 
 only depends on~$n$ and~$L$.

\begin{remark}\label{remarkgn}
We have the following conjectural statement, which is stronger than~\eqref{Chatd1dnuniformind}, that
for any fixed $L\ge0$, 
\begin{align}\label{Chatd1dnuniformindn1118}
\widehat{C}(\dd) \= \sum_{k=0}^{L-1}\frac{\widehat{c}_k(p_1(\dd),p_2(\dd),\dots)}{(X(\mathbf{d})-p_0(\dd))^k} 
\+ O\biggl(\frac{|\widehat{c}_{L}(p_1(\dd),p_2(\dd),\dots)|+1}{(X(\mathbf{d})-p_0(\dd))^{L}}\biggr)\,,
\end{align}
as $g(\dd)\to\infty$,
where the implied constant only depends on~$L$.
In terms of~$C(\dd)$, this conjecture states that for any fixed $L\ge0$, 
\begin{align}\label{Chatd1dnuniformindn0808}
C(\dd)= \frac{1}{\pi}\sum_{k=0}^{L-1}\frac{c_k(p_1(\dd),p_2(\dd),\dots)}{(X(\mathbf{d})-p_0(\dd))^k} 
\+ O\biggl(\frac{|c_{L}(p_1(\dd),p_2(\dd),\dots)|+1}{(X(\mathbf{d})-p_0(\dd))^{L}}\biggr) \,,
\end{align}
as $g(\dd)\to\infty$, where the implied constant only depends on~$L$.
\end{remark}

We end this section by giving some information about the interval $I_{g,n}$ defined in the introduction.
Denote by $m(g,n)$ and $M(g,n)$ its endpoints, i.e., the minimum and maximum of all $C(\dd)$ with 
$\dd\in(\mathbb{Z}_{\ge1})^n$ and $g=|\dd|+1$. Conjecture~\ref{conjmonotoncity} implies
that $m(g,n)=C(1^{n-1},g-n)$ and $M(g,n)=C(d^p (d+1)^{n-p})$,
where $d =[\frac{g-1}n]$ and $p=(d+1)n-g+1$. Using~\eqref{Chatd1dnuniformindn1118}, we then find the conjectural asymptotic formulas
\begin{align}
& \gamma(2g-2+n) \m m(g,n) \= \frac{27\thin n}{8\thin\pi\,(2g-2+n)^4} \+ O\Bigl(\frac1{g^4}\Bigr)\,,  \label{asymgn}\\
& \gamma(2g-2+n) \m M(g,n) \= \frac{(2d+1)!!^3}{2^{d+1}\thin\pi\,(d+1)!}\,\frac{(d+1)\thin n-g}{(2g-2+n)^{2d+2}} 
    \+ O_L\Bigl(\frac1{g^{\min\{2d+2, L\}}}\Bigr)
	\label{asyMgn}\end{align} 
for any fixed $L$,
where $d =[\frac{g-1}n]$.
In particular, when $d=[\frac{g-1}n]$ is fixed, $M(g,n)$ differs from $\gamma(2g-2+n)$ 
by a quantity of the order of $g^{-2d-1}$. 
We are currently trying to prove some version of the above conjectural asymptotics
(which greatly refines Theorem~\ref{thmuniformasymptotics}) and hope to return to this later.

\section{Further asymptotic formulas}\label{sectwopoint}
In this section, we will 
give asymptotic formulas of another type for BGW numbers, including and based on a closed 
asymptotic formula for two-point BGW numbers. 

For every $d\ge1$, define the power series $W_d(X)\in X^{-2d-2}\,\mathbb{Q}[[X^{-1}]]$ by the formula 
\begin{align}\label{defWdX0702}
	1-\widehat{C}(d,\tfrac{X}2-1-d) \;\sim\; W_d(X)\,, \qquad {\rm as}~X\to\infty\,.
\end{align}
(The fact that $W_d(X)$ is well defined is because of Theorem~\ref{thmpoly}.)
In terms of the polynomials $\widehat{c}_k(\mathbf{p})$, we have
\begin{align}\label{defWdX0628}
W_d(X) \= -\sum_{k=2d+2}^{\infty} \frac{\widehat{c}_k(\mathbf{e}_d)}{X^k}\,.
\end{align}

Formula~\eqref{TwoPointsC} implies that, for every $d\ge2$,
\begin{align}\label{recWd}
W_{d-1}(X)-W_{d}(X) \= \frac{(2d-1)!!^3}{8^d \,d!} \,\frac{(\tfrac{X}2-2d)\,\Gamma\bigl(\frac{X+3}2\bigr)\,\Gamma\bigl(\frac{X+1}{2}-d\bigr)^3}{\Gamma\bigl(\frac{X}2+1\bigr)^3\, \Gamma\bigl(\frac{X}2+1-d\bigr)}\,, 
\end{align}
where the right-hand side is interpreted as a power series of $X^{-1}$ by Stirling's formula.  Together with the limiting condition $W_{\infty}(X)=0$, formula~\eqref{recWd} determines $W_d(X)$, $d\ge1$, completely, the first few terms being
\begin{align}
W_d(X) &\=  \frac{(2d+1)!!^3}{2^{d+1}\, (d+1)!}\,\Bigl(\frac{1}{X^{2d+2}}+\frac{(2d-1)(d+1)}{X^{2d+3}} \nn\\
&\qquad +\,\frac{(2d+3)(d+1)\bigl(6 d^3 + 7d^2 - 8 d+1\bigr)}{6 (d+2)\,X^{2d+4}}\,+\,\cdots\Bigr)\,.\label{WdXexpressioninpowers}
\end{align}

In the following proposition, we give an explicit formula for~$W_d(X)$.
\begin{prop}\label{propWdX0628}
The power series $W_d(X)$ for any fixed $d\ge1$ is given by
\beq W_d(X) \= \frac{(2d+1)!!^3}{2^{d+1}\, d!}  \sum_{j\ge1} \frac{A_j(d)}{d+j} \, \frac{1}{(X-2d-j+1)_{2d+2j}}\,.\label{expressionwdX0701}\eeq 
Here the inverse Pochhammer symbols $1/(X-2d-j+1)_{2d+2j}$ are interpreted as
elements of $Q[[1/X]]$, and $A_j(d)$ $(j\ge1)$ are polynomials defined by the asymptotic formula
\begin{align}
2^{-2d-4}\frac{\Gamma\big(\frac{X+1}2-d\big)^3\,\Gamma\big(\frac{X+3}2\big)}{\Gamma\big(\frac{X}2+2\big)^3\, \Gamma\big(\frac{X}2-d+1\big)} \;=\;  \sum_{j=1}^{\infty} \frac{A_j(d)}{(X-2d-j+1)_{2d+2j+2}}\,, \label{defAjd}
\end{align}  
in which both sides are interpreted as power series of~$X^{-1}$. More explicitly, 
	\begin{align}\label{Aj(d)0628}
		A_j(d) \= (-1)^{j-1}\,(j-1)!\, \sum_{0\leq l\leq [\frac{j+1}2]} \frac{(2l-1)!!}{8^{l}\, l!^3} \, (j-2l)_{2l} \; (d+\tfrac{3}2-l)_{j-1} \,.
	\end{align} 
\end{prop}

\begin{proof}
It is easy to verify that the left-hand side of~\eqref{defAjd} 
(as a power series in~$X^{-1}$) is invariant under $X\to2d-3-X$, so $A_j(d)$ is well defined from~\eqref{defAjd}.

Let us first use the mathematical induction to prove the following equality:
\begin{align}\label{TwoPointNew}
\widehat{C}(d_1,d_2+1)-\widehat{C}(d_1,d_2) \= \frac{4^{(d_1+d_2+2)}\, (2d_1+4d_2+7)\, (d_1+1) F_{d_1+1} \, F_{d_2+1}}{(2d_1+2d_2+5)! \; \gamma(2d_1+2d_2+4)}\,,
\end{align}
with $F_d$ and $\gamma(X)$ given by~\eqref{TwoPoints} and~\eqref{defgamma}, respectively. Denote by $G(d_1,d_2)$ the right-hand side of~\eqref{TwoPointNew}.
For $d_1=0$, \eqref{TwoPointNew} follows directly from~\eqref{smallest}. If we assume that~\eqref{TwoPointNew} is true for $d_1=k-1$, then for $d_1=k$
\begin{align*}
&\widehat{C}(k,d_2+1) - \widehat{C}(k,d_2) \= \big(\widehat{C}(k-1,d_2+2) - \widehat{C}(k-1,d_2+1)\big) \\
&\qquad \qquad \qquad + \big(\widehat{C}(k,d_2+1) - \widehat{C}(k-1,d_2+2)\big) -
\big(\widehat{C}(k,d_2) - \widehat{C}(k-1,d_2+1)\big)\nn\\
&= G(k-1,d_2+1) + \frac{4^{k+d_2+2}\, (d_2+2-k)\, F_{k}\, F_{d_2+2}}{(2k+2d_2+4)!} -\frac{4^{k+d_2+1}\, (d_2+1-k)\, F_{k}\, F_{d_2+1}}{(2k+2d_2+2)!} \nn\\
&= G(k,d_2)\,, \nn
\end{align*}
where for the second equality we used formula~\eqref{TwoPointsC}. This completes the proof of~\eqref{TwoPointNew}.

By~\eqref{defWdX0702} and~\eqref{TwoPointNew} we find the identity
\begin{align}\label{WdXdifference}
W_{d}(X) - W_d(X+2) \=  \frac{(2d+1)!!^3}{8^{d+1}\, d!} \, (X-d+\tfrac32) \, \frac{\Gamma\big(\frac{X+1}2-d\big)^3\,\Gamma\big(\frac{X+3}2\big)}{\Gamma\big(\frac{X+4}2\big)^3\, \Gamma\big(\frac{X+2}2-d\big)}\,,
\end{align}
where the right-hand side is understood as its asymptotic expansion in~$X^{-1}$ as $X\to\infty$. This formula together with $W_d(X)\in X^{-2d-2}\,\mathbb{Q}[[X^{-1}]]$ uniquely determines $W_d(X)$. It is easy to verify that 
the right-hand side of~\eqref{expressionwdX0701}, with $A_j(d)$ defined 
by~\eqref{defAjd}, has the same recursive property. Hence the
first statement of Proposition~\ref{propWdX0628} is proved.

Let us now prove~\eqref{Aj(d)0628}. Denote by $r_d(X)\in\mathbb{Q}[[X^{-1}]]$ the asymptotic expansion of the left-hand side of~\eqref{defAjd}. Using the property $\Gamma(z+1)=z\Gamma(z)$, we see that
\begin{align}\label{rdXrecursion}
(X-2d-1)^3\,r_{d+1}(X) \= (X-2d)\, r_d(X).
\end{align}
From~\eqref{defAjd} and~\eqref{rdXrecursion}, we obtain
the following two recursions for $A_j(d)$:
\begin{align}
&A_{j+1}(d+1) - A_{j+1}(d) \= 2 (d+j+1) (2 d+j+2)\, A_j(d) \nn\\
&\qquad \qquad \qquad \qquad \qquad -((2d+j+3)(2d+3)+j^2)\, A_j(d+1)\,, \label{Ajrec2}\\
&-j^3\, A_{j}(d+1) + (2 d-j+3) \, A_{j+1}(d+1)-(2d+j+3) \,  A_{j+1}(d)\=0\,.\label{recAjd}
\end{align}
Here $j\ge0$, and we make the convention that $A_0(d)\equiv0$.
Notice that equations~\eqref{Ajrec2}--\eqref{recAjd}, together with the initial value $A_1(d)\equiv1$, 
uniquely determine all $A_j(d)$. It is easy to verify that the right-hand side of~\eqref{Aj(d)0628} satisfies~\eqref{Ajrec2}--\eqref{recAjd} and takes value~$1$ when $j=1$.
This completes the proof of~\eqref{Aj(d)0628}. 
\end{proof}
It follows from~\eqref{defWdX0628} and~\eqref{expressionwdX0701} that 
\begin{align}\label{chatked}
\widehat{c}_k(\mathbf{e}_d) \= -\frac{(2d+1)!!^3}{2^{d+1}d!} \,  \sum_{j=1}^{[\frac{k}2]-d} \frac{A_j(d)}{d+j} \,\sum_{l=0}^{k-2d-2j}\binom{k-1}{l}(-j)^l\, S(k-l-1,2d+2j-1)\,,
\end{align}
where $S(n,k)$ are the Stirling numbers of the second kind.

It is interesting to notice that the polynomial $A_j(d)$ is the product of $(d+3/2)_{[j/2]}$ and a polynomial of degree $[(j-1)/2]$, which can be easily proved by using~\eqref{Aj(d)0628}. For the reader's convenience we provide the first few $A_j(d)$:
\begin{align}
&A_1(d) \=1\,,\quad A_2(d) \= -\frac12 \,(2d+3)\,, \quad A_3(d)\= \frac18 \,(2d+3)\,(10d+21)\,, \nn\\
&A_4(d)\= -\frac{3}{16} (2d+3)\,(2d+5)\, (14d+31)\,.\nn
%A_5(d)\= \frac{3}{128} (2 d+3) (2 d+5) \left(676 d^2+3584 d+4695\right)\nn
\end{align}

We also remark that although $W_d(X)$ is defined as the asymptotics of two-point BGW numbers it also gives information about multi-point BGW numbers. Indeed, from~\eqref{Chatexpansion} and Theorem~\ref{thmpoly} we can deduce that
for a given $n\ge2$ and for $\dd=(d_1,\dots,d_{n})$ with $1\leq d_1\le \cdots\le d_{n-1}$ fixed,
\begin{align}
	\widehat{C}(\dd) \= 1\,-\,\sum_{j=1}^{n-1}W_{d_j}(X(\dd))+O(X(\dd)^{-2d_1-2d_2-3})\,,
	\qquad d_n\to\infty\,.\label{ChatdexpansionbyWd}
\end{align} 
 
Similar to~\eqref{defWdX0702}, for $n\ge1$, $\lambda=(\lambda_1,\dots,\lambda_{n-1})$, define $W_{\bl}(X)\in \mathbb{Q}[[X^{-1}]]$ via 
\begin{align}\label{Chatasymptoticsgeneral}
\widehat{C}(\lambda,d_n)  \;\sim\; -\sum_{I\subset\{1,\dots,n-1\}}  W_{\bl_I}(X(\lambda,d_n) )\,, \qquad d_n\to\infty\,,
\end{align}
with $W_{\emptyset}(X)=-1$. Then we have the following proposition.
\begin{prop}\label{propWlambdaX}
We have
\begin{align}\label{chatkWlambdaX}
\sum_{k\ge0}\frac{\widehat{c}_k(\mathbf{p})}{X^k}\= -\sum_{q_1,q_2,\dots\ge0} W_{1^{q_1},2^{q_2},\dots}(X)\,
\prod_{i=1}^{\infty}\binom{p_i}{q_i}\,,
\end{align}
where both sides are understood as elements in~$\mathbb{Q}[p_1,p_2,\dots][[X^{-1}]]$.
Moreover, the power series $W_{\bl}(X)\in X^{-2|\lambda|-\ell(\lambda)-1}\,\mathbb{Q}[[X^{-1}]]$.
\end{prop}
\begin{proof}
Using~\eqref{Chatexpansion}, Theorem~\ref{thmpoly} and~\eqref{Chatasymptoticsgeneral}, we obtain~\eqref{chatkWlambdaX}.
Using~\eqref{degreeassign}--\eqref{degreeest} and comparing degrees on both sides of~\eqref{chatkWlambdaX}, 
we obtain that $W_{\bl}(X)\in X^{-2|\lambda|-\ell(\lambda)-1}\,\mathbb{Q}[[X^{-1}]]$.
\end{proof}
We list a few examples of $W_{\lambda}(X)$ below:
\begin{align}
&W_{1,1}(X) \= \frac{1701}{4X^7} + \frac{380295}{64X^8} + \frac{832815}{16X^9}+ \frac{2935197}{8X^{10}}+\cdots\,, \nn\\
&W_{1,2}(X) \= \frac{388125}{16X^9} + \frac{83804625}{128X^{10}} + \frac{1336975875}{128X^{11}} + \frac{131751025875}{1024X^{12}}+\cdots\,,  \nn\\
&W_{1,1,1}(X) \= \frac{1754703}{8X^{10}} + \frac{245520639}{32X^{11}} + \frac{79688662083}{512X^{12}} + \frac{615031348329}{256X^{13}}+\cdots\,,  \nn\\
&W_{1,1,2}(X) \= \frac{779513625}{32X^{12}} + \frac{21043769625}{16X^{13}} + \frac{41031922798125}{1024X^{14}}+\cdots\,. \nn
%&W_{1,1,1,1}(X) \= \frac{1172710953}{4X^{13}} + \frac{76212197649}{4X^{14}} + \frac{88527339433503}{128X^{15}} +\cdots\,.\nn
\end{align}

We note that Proposition~\ref{propWlambdaX} together with~\eqref{Chatexpansion} and Theorem~\ref{thmpoly} implies~\eqref{ChatdexpansionbyWd} and formulas like 
\begin{align}\label{ChatdexpansionbyWd1d2}
\widehat{C}(\dd)\=1-\,\sum_{i=1}^{n-1}W_{d_i}(X(\dd))-\sum_{1\leq i<j\leq n-1}W_{d_i,d_j}(X(\dd))+O(X(\dd)^{-2d_1-2d_2-2d_3-4})\,,
\end{align}
for $n\ge1$, $1\le d_1\le \cdots \le d_{n-1}$ fixed and $d_n\to\infty$. 
Based on numerical experiments, we conjecture that the leading term of $W_{\bl}(X)$ is 
\begin{align}
	2 \, (2|\bl|+\ell(\lambda))!\, C(\bl)X^{-2|\bl|-\ell(\lambda)-1}\,.
\end{align}

\section{Subexponential asymptotic terms} \label{secsubleadingasymptotics}
In Section~\ref{secintro} we have described the discovery of the conjectural asymptotic formula~\eqref{dddasympleading}.
In this section we study subexponential asymptotics more systematically. 

Let us look, for a fixed partition $\bmu=(1\leq \mu_1\leq\dots\leq\mu_n)$, at the asymptotics of $1-\widehat{C}(\dd)$ with $\dd=\bmu\thin d=(\mu_1\thin d,\dots,\mu_n\thin d)$ as $d\to\infty$. 
Similar to~\eqref{dddasympleading}, we find (based on computations using formula~\eqref{formulaBGWeq2}) the following conjectural asymptotic formulas as $d\to\infty$:
\begin{align*}
&1-\widehat{C}(\mu_1d,\mu_2d) \,\sim\, \frac{2}{\sqrt{\pi\,d}}\,
\frac{\mu_1^{2\mu_1d+3/2}\,\mu_2^{2\mu_2d+3/2}}{|\bmu|^{2|\bmu|d+7/2}}\,,\nn\\
&1-\widehat{C}(\mu_1d,\mu_2d,\mu_3d)\,\sim\,\frac{2\,p_{\mu_1}}{\sqrt{\pi \,d}}\,
\frac{(\mu_2+\mu_3)^{2(\mu_2+\mu_3)d+5/2}\,\mu_1^{2\mu_1d+3/2}}{|\bmu|^{2|\bmu|d+9/2}}\,,\nn\\
&1-\widehat{C}(\mu_1d,\mu_2d,\mu_3d,\mu_4d)\,\sim\,\frac{2\, p_{\mu_1}}{\sqrt{\pi\, d}}\,
	\frac{(\mu_2+\mu_3+\mu_4)^{2(\mu_2+\mu_3+\mu_4)d+7/2}\,\mu_1^{2\mu_1d+3/2}}{|\bmu|^{2|\bmu|d+11/2}}\,.\nn\\
\end{align*}
With the help of these formulas and based on more computations, we obtain the following conjectural asymptotic formula:
for any fixed $n\ge2$,

\begin{align}\label{larged}
1-\widehat{C}(d_1,\dots,d_n) \;\sim\; \Bigl(\frac{1}2\Bigr)^{\delta_{n,2}}\,\sum_{j=1}^n\frac{2}{\pi X(\dd)}\binom{X(\dd)}{2d_j+1}^{-1}\,, 
\quad \min_{1\leq j\leq n} \{d_j\}\to\infty\,.
\end{align}
Since 
$$\frac{2}{\pi} \, (2d+1)! \;\sim\; \frac{(2d+1)!!^3}{2^{d+1}\, (d+1)!}\,,\quad d\to\infty\,,$$ we can rewrite the conjectural formula~\eqref{larged} equivalently as 
\begin{align}\label{larged2}
&1-\widehat{C}(d_1,\dots,d_n) \;\sim\; \Bigl(\frac{1}2\Bigr)^{\delta_{n,2}} \sum_{j=1}^n \frac{(2d_j+1)!!^3}{2^{d_j+1}(d_j+1)!}\frac{1}{(X(\dd)-2d_j)_{2d_j+2}} \,,\quad \min_{1\leq j\leq n}\{d_j\}\to\infty\,.
\end{align}
\begin{remark}\label{remarklarged}
Notice that the form $\frac{(2d+1)!!^3}{2^{d+1}\, (d+1)!}$ in~\eqref{larged2} also appears in the leading coefficents of $W_d(X)$ in~\eqref{expressionwdX0701}, so we obtain from~\eqref{Chatd1dnuniformind} that formula~\eqref{larged2} holds true even if 
some $d_j$'s are not large (but requires $X(\dd)\to\infty$).
So we guess that
\begin{align}\label{1-chatduniform}
1-\widehat{C}(\dd) \= \Bigl(\sum_{j=1}^n \frac{(2d_j+1)!!^3}{2^{d_j+1}(d_j+1)!}\frac{1}{(X(\dd)-2d_j)_{2d_j+2}}\Bigr) \,(1+o(1))
\end{align}
for any fixed $n\ge3$ and for $\dd=(d_1,\dots,d_n)\in \bigl(\mathbb{Z}_{\ge0}\bigr)^n$, uniformly as $X(\mathbf{d})\to\infty$. 
The conjectural formula~\eqref{1-chatduniform} implies that, for any two partitions $\mathbf{d}$ and $\mathbf{d'}$ of the same length and the same weight, 
\begin{align}\label{asymptoticordering}
\widehat{C}(\mathbf{d})-\widehat{C}(\mathbf{d'}) \= \sum_{m\ge 0} \frac{(2m+1)!!^3}{2^{m+1}\, (m+1)!} \,\frac{p_m-p'_{m}}{(X(\dd)-2m)_{2m+2}}(1+o(1))\,, \quad 
\end{align}
as $X(\mathbf{d})=X(\mathbf{d'})\to\infty$.  
We note that there is a coherent consistence between formula~\eqref{asymptoticordering} and 
Conjecture~\ref{conjmonotoncity}. Another point is that formula~\eqref{asymptoticordering} 
also explains the phenomenon (cf.~\eqref{twonumbers}) described in Section~\ref{secintro}.
As a further example, we have
\begin{align*}
&\widehat{C}(2,3,14,19)=0.99999999969689849693814650552875212296\cdots\,,\\ &\widehat{C}(2,3,15,18)=0.99999999969689849693814752270899360002\cdots\,,
\end{align*}
whose difference is about $1.01718\times 10^{-24}$ 
and the prediction gives $0.92432\times 10^{-24}$ with error  less than 10 percent.
\end{remark}
Now we compute the subleading terms for the subexponential asymptotics~\eqref{dddasympleading},~\eqref{larged2}. Based on the numerical experiments, we find that the error between 1 and $\widehat{C}(d^n)$ for fixed $n\ge2$ is of the form
\begin{align}\label{dddasymp}
1-\widehat{C}(d^n) \,\sim\, Y_n(d)\Bigl(1+\frac{b_1(n)}{d}+\frac{b_2(n)}{d^2}+\cdots\Bigr)\,,\qquad d\to\infty\,,
\end{align}
where
\begin{align}
Y_n(d) \= \Bigl(\frac{1}2\Bigr)^{\delta_{n,2}}\,
\sqrt{\frac{4(n-1)}{\pi\,n\,d}}\biggl(\frac{(n-1)^{n-1}}{n^n}\biggr)^{2d+1}\,,
\end{align}
and $b_1(n), b_2(n),\dots$ are rational number whose numerical values become a little simpler if we set
\begin{align}\label{defLn}
L_n \: 24\,\log\bigl(1+b_1(n)x+b_2(n)x^2+\cdots\bigr)\,,
\end{align}
in which the first few values are given by
\begin{align}
&L_n \= \frac{-11n^2 - n + 7}{n^2 - n}\, x + \frac{14 n^3 - 16 n^2 + n + 7}{2 n (n-1)^2} \, x^2 \nn\\
&\qquad \qquad + \frac{-721 n^6 + 1803 n^5 - 1953 n^4 + 901 n^3 - 243 n^2 + 93 n - 31}{120n^3(n-1)^3} \, x^3 +\cdots \,.\label{expressionLn}
\end{align}
Recall the discovery in Section~\ref{sectwopoint} that the asymptotics for two-point BGW numbers are building blocks for the higher-point numbers, at least when all but one $d_j$'s are fixed. The following observation generalizes this for the subexponential asymptotics. 

\smallskip

\noindent {\it Observation.}	
First of all, we find that the ratio of $\widehat{C}(d,d,d,d)-1$ and $\widehat{C}(d,3d+1)-1$, which are both exponentially small, is asymptotically equal to~4 to all orders in $d^{-1}$. Here the number $3d+1$ is chosen such that these two BGW numbers have the same argument~~$X(\mathbf{d})$. 
More generally, we conjecture that the ratio of $\widehat{C}(d^{n})-1$ and $\widehat{C}(d,d')-1$ is asymptotically equal to~$n$ for $n\ge3$, where $d'$ is defined by $2d'+1=(n-1)(2d+1)$. We notice that if $n$ is odd, then $d'$ is not an integer, but we can still define $C(d,d')$ by the two point formula~\eqref{TwoPointsC} with the definition of $F_h$ there replaced by 
\beq F_h \: \frac{\Gamma\bigl(h+\frac12\bigr)^3}{\pi^{3/2}\, \Gamma(h+1)}\,. \label{defnewFh}\eeq 

Now let us focus on the asymptotics for $1-\widehat{C}(d_1,d_2)$ again. Recall that if $d_1$ is fixed, this asymptotics is exactly $W_{d_1}(X)$ explicitly given in~\eqref{expressionwdX0701}. Another important observation is that,  the right-hand side of~\eqref{expressionwdX0701} makes sense (i.e. each term has lower order than the previous one) even if $X$ and $d$ are proportional and are in the region $X-2d\gg0$. 
Note that this property will not be true if we write $W_d(X)$ in another basis, e.g. the powers of $X^{-1}$ as in~\eqref{WdXexpressioninpowers} or the basis $1/(X-1)_{k}^{-}$ ($k=1,2,\dots$) used in~\cite{EGGGL}. Therefore, for any fixed integer $N\ge1$, we define a function $W(N;d,X)$ by  
\begin{align}
W(N;d,X) \= \frac{(2d+1)!!^3}{2^{d+1} d!} \sum_{j=1}^{N} \frac{A_j(d)}{d+j} \, \frac{1}{(X-2d-j+1)_{2d+2j}}\,,
\end{align}
where $A_j(d)$ are polynomials defined in~\eqref{Aj(d)0628}. We note that $W(N;d,X)$ is well defined in the region $X-2d>N$.
According to the previous observations, we conjecture that 
\begin{align}\label{twopointasymptoticB}
1-\widehat{C}(d_1,d_2) \,=\, W(N;d_1,X(d_1,d_2))\big(1+O(X(d_1,d_2)^{-N})\big)\,,\quad  X(d_1,d_2)\to\infty
\end{align}
holds for any fixed $N\ge1$ and for $d_1\leq d_2$. 
Now for general $\dd=(d_1,\dots,d_n)$ of length $n\ge3$, we have the following conjecture. 
\begin{conj}\label{conjWdX}
For any fixed $n\ge3$, any fixed $N\ge0$ and for $\dd=(d_1,\dots,d_n)$ satisfying 
$\min_{1\leq i<j\leq n}\{d_i+d_j\}\geq N/2$, we have the asymptotics:
\begin{align}\label{conjBdX}
1-\widehat{C}(\dd) \= \sum_{i=1}^{n}W(N;d_i,X(\dd))\, \bigl(1+O\bigl(X(\dd)^{-N}\bigr)\bigr)\,, \quad X(\mathbf{d})\to\infty\,.
\end{align}
\end{conj}

As a consequence of Conjecture~\ref{conjWdX}, for fixed $n\ge0$, the asymptotic expansion of $1-\widehat{C}(d^n)$ is explicitly given by
\begin{align}\label{asym1-Cd^n}
	1-\widehat{C}(d^n) \= n\, W(N;d,n(2d+1))\bigl(1+O(d^{-N})\bigr) \,, \quad d\to\infty\,,
\end{align}
for any fixed $N\ge0$. This explicitly gives all the numbers $b_1(n),b_2(n),\dots$ in~\eqref{dddasymp}, 
and the first three of them coincide with those in~\eqref{expressionLn}.

\section{Application to the Painlev\'e II hierarchy}
\label{secPainleve}
In this section, we will discuss the connections between the BGW numbers and two famous Painlev\'e
hierarchies. We begin with the Painlev\'e \uppercase\expandafter{\romannumeral 34} hierarchy (cf.~\cite{BR,CJP}),  by which
we mean
the following family of ODEs:
\begin{align}\label{p34}
2u \+ t \, \frac{du}{dt} \,-\, (2d+1) \, \frac{d}{dt} \Bigl(m_d\Bigl(u,\frac{du}{dt},\frac{d^2u}{dt^2},\dots,\frac{d^{2d}u}{dt^{2d}}\Bigr)\Bigr)\=0\,,  
\end{align}  
where $d\ge1$, and $m_d$ are the polynomials defined in~\eqref{defma},~\eqref{b(lambda)expansion}. 
The case with $d=1$ agrees with the Painlev\'e \uppercase\expandafter{\romannumeral 34} equation~\eqref{P34equation}.
%We have the following lemmas.
\begin{lemma}\label{lemP34sol}
For each $d\ge1$, there exists a unique formal solution to equation~\eqref{p34} of the form
\begin{align}\label{solpa}
u(t) \= \sum_{n\ge0} \frac{A_{d,n}}{t^{(2d+1)n+2}}\,, \qquad A_{d,0} \= \frac18\,.
\end{align}
\end{lemma}
\begin{proof}
If we assign degrees 
\beq \deg u_i=i+2\,, \quad i\ge0\,, \eeq
then the polynomials $m_d$ are homogeneous of degree $2d+2$. The lemma follows.
\end{proof}

\begin{lemma}\label{lem2}
The coefficients $A_{d,n}$ in~\eqref{solpa} are related to the BGW correlators by
\begin{align}\label{Adn}
A_{d,n} \=
\frac{((2d+1)n+1)!}{2^{2nd+1}\,(2d+1)!!^n\,n!} \,C(d^{n})\,, 
\end{align}
where $C(d_1,\dots,d_n)$ is defined in~\eqref{defC}.
\end{lemma}

\begin{proof}
As it was done in~\cite{BR}, dividing both sides of the $m=0$ case of~\eqref{LmZ=0} by~$Z$  and 
differentiating the resulting equality twice with respect to~$t_0$ and using~\eqref{KdV}, one obtains
\begin{align}\label{dEL}
2u-(1-t_0) \, u_{t_0}+\sum_{k\ge1}(2k+1)t_k \partial_{t_0} (m_k) \= 0 \,,
\end{align} 
where we recall that $u:= \partial_{t_0}^2(\log Z)$. Specializing ${\bf t}={\bf t}^*=(t_0^*,t_1^*,t_2^*,\dots)$ in~\eqref{dEL} with 
\begin{equation}\label{t*}
 t_d^*=1, \quad t_i^*=0\,( i\neq 0, \, i\neq d) \,,
\end{equation} 
we find
\begin{align}
2u(\mathbf{t^*})-(1-t_0)u_{t_0}(\mathbf{t^*})+
(2d+1)\partial_{t_0}(m_d(u(\mathbf{t^*}),u_{t_0}(\mathbf{t^*}),\dots))\=0\,.
\end{align}
The lemma is proved by noticing that
\begin{align}\label{u'}
u|_{\mathbf{t}=\mathbf{t^*}} \=\sum_{n\ge1}\frac{C(d^n)\,((2d+1)n+1)!}{2^{2nd+1}(1-t_0)^{(2d+1)n+2}(2d+1)!!^{n}n!} \+ \frac{1}{8(1-t_0)^2}
\end{align}
and by putting $t=1-t_0$.
\end{proof}

It is convenient to work with another 
normalization of the Painlev\'e XXXIV hierarchy:
\begin{align}
	2^{2d+1} (2d+1)!! \,\partial_X(m_d(Y/2, Y_X/2, Y_{XX}/2,\dots)) - XY_X-2Y \= 0 \,, \quad d\ge1\,, \label{newnormalizationp34}
\end{align}
which is related to~\eqref{p34} by the rescalings
\begin{equation} \label{rescale1028}
	t \= \frac12 \biggl(\frac{(2d-1)!!}{2}\biggr)^{-1/(2d+1)} X\,, \qquad u \= 2 \biggl(\frac{(2d-1)!!}{2}\biggr)^{2/(2d+1)} Y \,.
\end{equation}
The formal solution of interest (a solution to~\eqref{newnormalizationp34}) now has the form
\begin{align}\label{YX}
Y(X)
 \= \sum_{n\ge0}\frac{y_{d,n}}{X^{(2d+1)n+2}} \,, \qquad y_{d,0}=\frac14\,.
\end{align}
\begin{thm}\label{thmPainleveapporiginal}
For each $d\ge1$, the coefficients $y_{d,n}$ of the unique 
formal solution $Y(X)$ given in~\eqref{YX} to the Painlev\'e XXXIV hierarchy~\eqref{newnormalizationp34} have the following asymptotics:
\begin{align}
y_{d,n} \,\sim\, \frac{1}{\pi}\,\frac{((2d+1)n+1)!}{(2d+1)^nn!}\,, \qquad n\to\infty\,.\label{asymyak}
\end{align}
\end{thm}
\begin{proof}
By using~\eqref{Adn}, \eqref{rescale1028} and Theorem~\ref{thmuniformasymptotics} we obtain~\eqref{asymyak}.
\end{proof}
\begin{cor}\label{corasyydn}
For each $d\ge1$, we have
\begin{align}
y_{d,n} \,\sim\, \frac{1}{\pi}\,\frac{((2d+1)n+1)!}{ (2d+1)^nn!}\,\Bigl(1+\frac{r_1(d)}{n}+\frac{r_2(d)}{n^2}+\cdots\Bigr) \qquad (n\to\infty)\,, \label{fullasymyak}
\end{align}
with explicitly computable coefficients $r_k(d) \in \QQ$.
\end{cor}
\begin{proof}
Let us first show, without using Theorem~\ref{thmPainleveapporiginal}, that
\begin{align}
y_{d,n} \,\sim\, A\,\frac{((2d+1)n+1)!}{ (2d+1)^nn!}\,\Bigl(1+\frac{r_1(d)}{n}+\frac{r_2(d)}{n^2}+\cdots\Bigr)\,, \qquad n\to\infty\,, \label{fullasymyakA}
\end{align}
for some nonzero constant~$A$. 
Using~\eqref{newnormalizationp34},~\eqref{YX} and using the homogeneity of $m_d$, we obtain a recursion that expresses $(2d+1)\,n\,y_{d,n}$ by an element in
$${\rm span}_{\mathbb{Q}}\Bigl\{\sum_{n_1+\cdots+n_k=n-1}\prod_{j=1}^{k}(((2d+1)n_j+2)_{i_j}\,y_{d,n_j})\,\big|\, k\ge1, \mathbf{i}\in(\mathbb{Z}_{\ge0})^k, |\mathbf{i}|+2k=2d+3\Bigr\}\,.$$
The leading asymptotics in~\eqref{fullasymyakA} of~$y_{d,n}$ can be deduced from
the linear terms 
$$(2d+1)n\,y_{d,n} \m ((2d+1)(n-1)+2)_{2d+1}\; y_{d,n-1}\,.$$ 
It follows that
\begin{align}
\prod_{j=1}^{k}(y_{d,n_j}((2d+1)n_j+2)_{i_j}) &\= O\Bigl(y_{d,n}\frac{n!}{((2d+1)n+1)!}\prod_{j=1}^k \frac{((2d+1)n_j+1+i_j)!}{n_j!}\Bigr)\label{prodydnestimate}
\end{align}
for each $k\ge0$, $i_1+\cdots+i_k+2k=2d+2$ and $n_1+\cdots+n_k=n-1$. By using the logarithmic convexity of the function $((2d+1)n+1+i_j)!/n!$, we obtain the right-hand side of~\eqref{prodydnestimate} is $O(y_{d,n}n^{-2d(h-1)-2k+2})$ when $n_1,\dots,n_k\leq n-h$. This shows that up to any relative power in~$n^{-1}$,
the recursion that $y_{d,n}$ satisfies is linear and of finite order. This proves~\eqref{fullasymyakA}. The determination of~$A=1/\pi$ follows from Theorem~\ref{thmPainleveapporiginal}. 
\end{proof}

\begin{cor}\label{thmPainleve2}
For each $d\ge1$, we have the asymptotic expansion
\beq \widehat{C}(d^n) \;\sim\; 1 + \frac{a_1(d)}{X(d^n)} + \frac{a_2(d)}{X(d^n)^2} + \cdots  \qquad (n\to\infty) \label{Chatd^nasy} \eeq 
with explcitly computable $a_k(d)\in\mathbb{Q}$ , the first three cases being
\begin{align*}
&\widehat{C}(1^n) \;\sim\;1-\frac{9}{8 \,(3n)^3}-\frac{9}{4 \,(3n)^4}-\frac{219}{8 \,(3n)^5}+\cdots\,, \\
&\widehat{C}(2^n) \;\sim\; 1-\frac{225}{16 \,(5n)^5}-\frac{2025}{16 \,(5n)^6}-\frac{96075}{128 \,(5n)^7}+\cdots\,,\\
&\widehat{C}(3^n) \;\sim\; 1-\frac{55125}{128 \,(7n)^7}-\frac{275625}{32 \,(7n)^8}-\frac{3340575}{32 \,(7n)^9}+\cdots\,.
\end{align*}
\end{cor}
\begin{proof}
Formula~\eqref{Chatd^nasy} is proved by using~\eqref{normalizehatc61}, \eqref{Adn}, \eqref{rescale1028}, \eqref{YX}, Corollary~\ref{corasyydn} and Stirling's formula. 
\end{proof}
 
By~\eqref{chatked} and the conjectural formula~\eqref{Chatd1dnuniformindn0808}, one can deduce that for any fixed $d\ge1$,
\beq 1-
\widehat{C}(d^n) \;\sim\;  \frac{(2d+1)!!^2\thin (2d-1)!!}{2^{d+1}\thin(d+1)!\,X(d^n)^{2d+1}} \qquad  (\thin X(d^n) = (2d+1)n\thin)
\eeq 
to leading order as $n\to\infty$. 
This conjectural formula is consistent with Corollary~\ref{thmPainleve2}. Compare also with 
the asymptotic formula~\eqref{dddasympleading} when $n$ is fixed and $d$ tends to infinity. 

\begin{proof}[Proof of Theorem~\ref{thmPainleveapp}]
According to~\cite{CJP} (cf.~\cite{FA}), performing the following invertible transformation 
\begin{align}
	&Y=V_X-V^2, \label{VtoY}\\
	&V=-\frac{2^{2d-1} (2d-1)!! \, \partial_X(
		m_{d-1}(\frac Y2, \frac{Y_X}2,\dots))-\alpha_d}{2^{2d} (2d-1)!! \,
		m_{d-1}(\frac Y2, \frac{Y_X}2,\dots)-X} \label{YtoV}
\end{align}
on the Painlev\'e~XXXIV hierarchy yields the Painlev\'e II hierarchy~\eqref{P2hier}.

We note that in general, $\alpha_d$ could be an arbitrary constant. But the particular solution~$V(X)$
derived by the above transformation of the power series in~\eqref{YX} only solves the Painlev\'e~II hierarchy for the 
parameter~$\alpha_d=\frac{1}2$. This can be seen by comparing the coefficients of~$X^{-2}$ on both sides of~\eqref{VtoY}, 
and by noticing that $V(X)$ has the leading term $-\alpha_d/X$. 
So $Y(X)$ defined in~\eqref{YX} corresponds to the  formal solution~\eqref{solP2hier} to~\eqref{P2hier} with $\alpha_d=1/2$.

To prove formula~\eqref{asymvak}, we note that the transformation~\eqref{VtoY} gives the following relations between $v_{d,n}$ and $y_{d,n}$:
\begin{align}\label{ydnvdn}
y_{d,n} \= ((2d+1)n+1)\, v_{d,n} \,-\, \sum_{n_1+n_2=n} \, v_{d,n_1}\,v_{d,n_2}\,.
\end{align}
Using the asymptotics~\eqref{asymyak} of $y_{d,n}$ and facts about asymptotics of very rapidly divergent series (cf.~\cite{CMZ}), it is easy to show that $v_{d,n}$ is asymptotically equal to 
$y_{d,n}/((2d+1)n-1)$,  as claimed in~\eqref{asymvak}.
\end{proof}

Similar to Corollary~\ref{corasyydn}, we have the following
\begin{cor}
For fixed $d\ge1$,  the coefficients $v_{d,n}$ of the formal solution~\eqref{solP2hier} to the Painlev\'e~II hierarchy have the following asymptotic expansion:
\begin{align}
v_{d,n} \,\sim\, \frac{1}{\pi} \, \frac{((2d+1)n-1)!}{(2d+1)^{n-1} (n-1)!} \,\Bigl(1+\frac{s_1(d)}{n}+\frac{s_2(d)}{n^2}+\cdots\Bigr)\,, \qquad n\to\infty\,,\label{asymvdn}
\end{align}
with explicitly computable coefficients $s_k(d)\in \QQ$.
\end{cor}
\begin{proof}
Similar to the proof of Corollary~\ref{corasyydn}.
\end{proof}

\section{Application to BGW-kappa numbers} \label{seckappaclass} 
For $g,m\ge1$, $n\ge0$, $d_1,\dots,d_n\ge0$, define the BGW-kappa numbers by
\beq\label{defkappa}
\Bigl\langle\kappa_1^m\prod_{j=1}^n\tau_{d_j}\Bigr\rangle_g^{\Theta} \;:=\; 
\int_{\overline{\mathcal{M}}_{g,n}} \, \kappa_1^m \, \psi_1^{d_1}\cdots\psi_n^{d_n} \, \Theta_{g,n} \,,
\eeq
which vanishes unless $m+d_1+\cdots+d_n=g-1$. One could of course also include powers of other kappa's in~\eqref{defkappa}, but we will study only the integrals with a power of $\kappa_1$ and here we use the terminology ``BGW-kappa numbers" for convenience. Another name that can be found in the literature is super JT gravity. 
According to~\cite{Norbury}, these numbers are related to the volumes of moduli spaces of super hyperbolic surfaces, called Stanford--Witten volumes $V_{g,n}^{\Theta}(\mathbf{L})$. This relation is given by
\begin{align}\label{SW-kappa}
V_{g,n}^{\Theta}(L_1,\dots,L_n)\=\sum_{m,d_1,\dots,d_n\ge0\atop m+d_1+\cdots+d_n=g-1}\langle\kappa_{1}^{m}\tau_{d_1}\cdots\tau_{d_n}\rangle_g^{\Theta}\frac{(2\pi^2)^{m}}{m!}\prod_{j=1}^{n}\frac{L_j^{2d_j}}{2^{d_j}d_j!}.
\end{align}

As proved in~\cite{Norbury} (cf.~\cite{KMZ,LX3,MZ}), the BGW-kappa numbers can be expressed in terms of the BGW numbers as follows:
\begin{align}
&\langle\kappa_1^m\prod_{j=1}^n\tau_{d_j}\rangle^{\Theta}_g
\=\sum_{l=1}^{m}\frac{(-1)^{m-l}}{l!}\sum_{m_1,\dots,m_l\ge1\atop m_1+\cdots+m_l=m}
\binom{m}{m_1,\dots,m_l}\biggl\langle\prod_{j=1}^n\tau_{d_j}\prod_{i=1}^l\tau_{m_i}\biggr\rangle^{\Theta}_g\,.\label{KMZ}
\end{align}
Using~\eqref{KMZ} we compute a few BGW-kappa numbers in Table~\ref{tablekappaBGW}. 
%We conjecture this is true in general. Actually, 
It follows from~\eqref{KMZ} and the Integrality Conjecture for the BGW numbers that 
the BGW-kappa numbers are integral away from powers of~$2$. We refer to the latter statement as the Integrality Conjecture for the 
BGW-kappa numbers.

\begin{table}[h!]
\begin{center}
\begin{tabular}{|l|l|}
\hline
\multicolumn{2}{|c|}{$g=2$}\\
\hline
$\langle\kappa_1\rangle^{\Theta}=\frac{3}{128}\approx2.34\times10^{-2}$&\\
\hline
\multicolumn{2}{|c|}{$g=3$}\\
\hline
$\langle\kappa_1\tau_{1}\rangle^{\Theta}=\frac{63}{512}\approx1.23\times10^{-1}$&$\langle\kappa_1^2\rangle^{\Theta}=\frac{111}{1024}\approx1.08\times10^{-1}$\\
\hline
\multicolumn{2}{|c|}{$g=4$}\\
\hline
$\langle\kappa_{1}\tau_{1}^2\rangle^{\Theta}=\frac{7221}{2048}\approx3.53$
&$\langle\kappa_1^2\tau_{1}\rangle^{\Theta}=\frac{106911}{32768}\approx3.26$
\\ $\langle\kappa_1\tau_{2}\rangle^{\Theta}=\frac{8625}{32768}\approx2.63\times10^{-1}$
&$\langle\kappa_1^3\rangle^{\Theta}=\frac{45093}{16384}\approx2.75$\\
\hline
\multicolumn{2}{|c|}{$g=5$}\\
\hline
$\langle\kappa_1\tau_{1}^3\rangle^{\Theta}=\frac{4825971}{16384}\approx2.95\times10^{2}$ 
& $\langle\kappa_1^2\tau_2\rangle^{\Theta}=\frac{1974135}{131072}\approx1.51\times10$ \\ 
$\langle\kappa_1\tau_1\tau_2\rangle^{\Theta}=\frac{524925}{32768}\approx1.60\times10$ 
&$\langle\kappa_1^3\tau_1\rangle^{\Theta}=\frac{16199169}{65536}\approx2.47\times10^2$\\
$\langle\kappa_{1}\tau_{3}\rangle^{\Theta}=\frac{44835}{65536}\approx6.84\times10^{-1}$  
& $\langle\kappa_1^4\rangle^{\Theta}=\frac{53483271}{262144}\approx2.04\times10^2$  \\
$\langle\kappa_1^2\tau_1^2\rangle^{\Theta}=\frac{9127017}{32768}\approx2.79\times10^2$ &  \\
\hline
\multicolumn{2}{|c|}{$g=6$}\\
\hline
$\langle\kappa_{1}\tau_1^4\rangle^{\Theta}=\frac{3540311739}{65536}\approx5.40\times10^4$
&		$\langle\kappa_{1}^2\tau_1\tau_2\rangle^{\Theta}=\frac{1155623625}{524288}\approx2.20\times10^3$ \\
$\langle\kappa_{1}\tau_1^2\tau_2\rangle^{\Theta}=\frac{605705625}{262144}\approx2.31\times10^3$
& $\langle\kappa_{1}^2\tau_3\rangle^{\Theta}=\frac{151428375}{2097152}\approx7.22\times10$\\
$\langle\kappa_{1}\tau_{2}^2\rangle^{\Theta}=\frac{55787625}{524288}\approx1.06\times10^2$ 
& $\langle\kappa_1^3\tau_1^2\rangle^{\Theta}=\frac{386376633}{8192}\approx4.72\times10^4$ \\ 
$\langle\kappa_{1}\tau_1\tau_3\rangle^{\Theta}=\frac{19922175}{262144}\approx7.60\times10$
&	$\langle\kappa_1^3\tau_2\rangle^{\Theta}=\frac{4184142525}{2097152}\approx2.00\times10^3$ \\
$\langle\kappa_{1}\tau_{4}\rangle^{\Theta}=\frac{8831025}{4194304}\approx2.11$
& $\langle\kappa_1^4\tau_1\rangle^{\Theta}=\frac{171037302471}{4194304}\approx4.08\times10^4$ \\ 
$\langle\kappa_{1}^2\tau_1^3\rangle^{\Theta}=\frac{13555541331}{262144}\approx5.17\times10^4$
& $\langle\kappa_1^5\rangle^{\Theta}=\frac{69673098483}{2097152}\approx3.32\times10^4$\\
\hline
\end{tabular}
\end{center}
\caption{BGW-kappa numbers up to genus~6}
\label{tablekappaBGW}
\end{table}

We also provide a table (see Table~\ref{tablenormalizedkappanumbers}) for the normalized BGW-kappa numbers $C(m;\mathbf{d})$ defined in~\eqref{defCkappa} with $g\leq7$. The data for $C(1;\dd)$ is omitted since $C(1;\dd)=C(1,\dd)$. As before, we have listed in Table~\ref{tablenormalizedkappanumbers} the smallest common denominator $D=D_g$ of these numbers for each~$g$ and then tabulated the integer $D\, C(m;\dd)$ in the last column.
Inspired by Conjecture~\ref{conjmonotoncity}, and based on Table~\ref{tablenormalizedkappanumbers} and further numerical experiments,  we also conjecture that the function $(m,\dd)\mapsto C(m;\dd)$ for partitions $(m,\dd)$ of $g-1$ is monotone with respect to the ordering that
$(m,\dd)\prec (m',\dd')$ if $m>m'$, or $m=m'$ and $\dd\prec \dd'$. 

\begin{table}[phbt]
\begin{center}
\begin{tabular}{|l|c|l|l|}
\hlinew{1.5pt}
\multicolumn{4}{|c|}{$g=3,\quad D=1280$}\\
\hlinew{1.5pt}
$C(2;\emptyset)$&$\frac{333}{1280}$&$0.260156$&$333$\\
\hlinew{1.5pt}
\multicolumn{4}{|c|}{$g=4,\quad D=1146880$}\\
\hlinew{1.5pt}
$C(3;\emptyset)$&$\frac{135279}{573440}$&$0.235908$&$270558$\\
\hline
$C(2;1)$&$\frac{45819}{163840}$&$0.279657$&$320733$\\
\hlinew{1.5pt}
\multicolumn{4}{|c|}{$g=5,\quad D=252313600$}\\
\hlinew{1.5pt}
$C(4;\emptyset)$&$\frac{53483271}{252313600}$&$0.211971$&$53483271$\\
\hline
$C(3;1)$&$\frac{2314167}{9011200}$&$0.256810$&$64796676$\\
\hline
$C(2;2)$&$\frac{131609}{458752}$&$0.286885$&$72384950$\\
$C(2;1,1)$&$\frac{9127017}{31539200}$&$0.289386$&$73016136$\\
\hlinew{1.5pt}
\multicolumn{4}{|c|}{$g=6,\quad D=734737203200$}\\
\hlinew{1.5pt}
$C(5;\emptyset)$&$\frac{69673098483}{367368601600}$&$0.189654$&$139346196966$\\
\hline
$C(4;1)$&$\frac{24433900353}{104962457600}$&$0.232787$&$171037302471$\\
\hline
$C(3;2)$&$\frac{278942835}{1049624576}$&$0.265755$&$195259984500$\\
$C(3;1,1)$&$\frac{386376633}{1435033600}$&$0.269246$&$197824836096$\\
\hline
$C(2;3)$&$\frac{3365075}{11534336}$&$0.291744$&$214355277500$\\
$C(2;1,2)$&$\frac{5926275}{20185088}$&$0.293597$&$215716410000$\\
$C(2;1,1,1)$&$\frac{13555541331}{45921075200}$&$0.295192$&$216888661296$\\
\hlinew{1.5pt}
\multicolumn{4}{|c|}{$g=7,\quad D=399697038540800$}\\
\hlinew{1.5pt}
$C(6;\emptyset)$&$\frac{1057428386631}{6245266227200}$&$0.169317$&$67675416744384$\\
\hline
$C(5;1)$&$\frac{1196989428069}{5709957693440}$&$0.209632$&$83789259964830$\\
\hline
$C(4;2)$&$\frac{103748833683}{427483463680}$&$0.242697$&$97005159493605$\\
$C(4;1,1)$&$\frac{2242040330133}{9084023603200}$&$0.246811$&$98649774525852$\\
\hline
$C(3;3)$&$\frac{31418131}{115343360}$&$0.272388$&$108872620991680$\\
$C(3;1,2)$&$\frac{80848213893}{293894881280}$&$0.275092$&$109953570894480$\\
$C(3;1,1,1)$&$\frac{6931945897497}{24981064908800}$&$0.277488$&$110911134359952$\\
\hline
$C(2;4)$&$\frac{354207573}{1199570944}$&$0.295279$&$118021963323600$\\
$C(2;1,3)$&$\frac{222438209}{749731840}$&$0.296690$&$118586257982080$\\
$C(2;2,2)$&$\frac{4360002121}{14694744064}$&$0.296705$&$118592057691200$\\
$C(2;1,1,2)$&$\frac{3184112229}{10687086592}$&$0.297940$&$119085797364600$\\
$C(2;1,1,1,1)$&$\frac{466903889307}{1561316556800}$&$0.299045$&$119527395662592$\\
			\hline
		\end{tabular}\vskip 7pt
	\end{center}
	\caption{Numerical data for $C(m,\dd)$ with $g\le 7$} \label{tablenormalizedkappanumbers}
\end{table}

We now give a proof of Proposition~\ref{corkappa}. 
\begin{proof}[Proof of Proposition~\ref{corkappa}]
If follows from~\eqref{KMZ} and~\eqref{defCkappa} that 
\begin{align}
&C(m;\dd)
\= 3^m\,\sum_{l=1}^{m}\frac{(-1)^{m-l}}{l!}\sum_{m_1,\dots,m_l\ge1\atop m_1+\cdots+m_l=m}
\binom{m}{m_1,\dots,m_l} \nn\\
&\qquad\qquad\qquad\qquad \times \frac{C(d_1,\dots,d_n,m_1,\dots,m_l)}{(X(m;\dd)-m+l)_{m-l}\,\prod_{i=1}^l (2m_i+1)!!}\, \label{Cmd}.
\end{align}
For fixed $m\ge0$, the summation in the RHS of~\eqref{Cmd} is finite. The summand corresponding to $l=m$, $m_1=\cdots=m_l=1$ contributes the leading term $C(\dd,1^m)$, which by Theorem~\ref{thmuniformasymptotics} equals $1/\pi+O(1/g(m;\dd))$ uniformly as $g(m;\dd)\to\infty$. The rest summands equal $O(1/(X(\mathbf{m;d})-m+l)_{m-l})$, which also equals $O(1/g(m;\dd))$ uniformly as $g(m;\dd)\to\infty$. This proves the proposition.
\end{proof}
For the special case when $n\ge1$ and $d_1,\dots,d_{n-1}\ge0$ are all fixed, 
Proposition~\ref{corkappa} is analogous to a  result given by Liu--Xu~\cite{LX} for Witten's intersection numbers.

\begin{remark}
In another direction, one can also study the large genus asymptotics of the numbers 
$\langle\kappa_1^{m}\tau_{d_1}\cdots\tau_{d_n}\rangle_g$ when $m$ is no longer fixed.
In particular, it was conjectured by 
Griguolo--Papalini--Russo--Seminara~\cite{GPRS} that $\langle\kappa_1^{m}\tau_{d_1}\cdots\tau_{d_n}\rangle_g$ 
has the following leading asymptotics for fixed $n$ and fixed $\dd=(d_1,\dots,d_n)\in(\mathbb{Z}_{\ge0})^n$ as $g\to\infty$:
\begin{align}\label{conjvol2024}
\langle\kappa_{1}^{g-1-|\dd|}\tau_{d_1}\cdots\tau_{d_n}\rangle_g^{\Theta} \;\sim\;
\frac{\pi^{2|\dd|+n-2}\, 2^{g-1-3|\dd|}}{3^{3g-\frac{7}2-|\dd|+n} \, \prod_{j=1}^{n}(2d_j+1)!!} \, (3g-4-|\dd|+n)!
\end{align}  
(see also~\cite{OS, SW} for the case when $n=1$; by numerical experiments we also found the conjectural asymptotics~\eqref{conjvol2024}),
with the first four cases being
\begin{align*}
&\langle\kappa_{1}^{g-1}\rangle_g^{\Theta} \;\sim\; \frac{2^{g-1}}{\pi^2 \, 3^{3g-\frac72}} \, (3g-4)!\,,\qquad
\langle\kappa_{1}^{g-2}\tau_1\rangle_g^{\Theta} \;\sim\; \frac{\pi\, 2^{g-4}}{3^{3g-\frac52}} \, (3g-4)!\,,\\
&\langle\kappa_{1}^{g-3}\tau_2\rangle_g^{\Theta} \;\sim\; \frac{\pi^3\, 2^{g-7}}{5\cdot3^{3g-\frac72}} \, (3g-5)!\,,\qquad 
\langle\kappa_{1}^{g-3}\tau_1^2\rangle_g^{\Theta} \;\sim\; \frac{\pi^4\, 2^{g-7}}{3^{3g-\frac32}} \, (3g-4)!\,.\nn
\end{align*}
%as $g\to\infty$. 
Recently, Huang~\cite{Huang} proved the Griguolo--Papalini--Russo--Seminara conjecture up to an absolute constant factor.
\end{remark}

The following proposition generalizes Proposition~\ref{proprationality}.
\begin{prop}\label{corkapparational}
For fixed $m\ge1$, fixed $n\ge2$, fixed $d_1,\dots,d_{n-1}\ge1$, and for $d_n$ being an indeterminate, we have
\begin{align}\label{Cmdrationality}
C(m;\dd) \=  \frac{1}{(X(m;\mathbf{d})-1)!} \, \frac{(2d_n+1)!!^3}{2^{d_n+1} \, d_n!}  \, 
P_{m;d_1,\dots,d_{n-1}}(X(m;\mathbf{d}))\,,
\end{align}
where $X(m;\dd)=X(\dd)+3m$ as before and 
$P_{m;d_1,\dots,d_{n-1}}(X)\in \mathbb{Q}[X]$. Moreover, $\frac{C(m;\mathbf{d})}{C(m+|d|)}$
is a rational function of~$X(m;\dd)$.
\end{prop}
\begin{proof}
Substituting~\eqref{UtoP} in~\eqref{Cmd}, we get~\eqref{Cmdrationality}. The second statement that $\frac{C(m;\mathbf{d})}{C(m+|d|)}$ is a rational function of $X(m;d)$ can then be deduced from~\eqref{smallest}. 
\end{proof}
The following proposition generalizes Theorem~\ref{thmpoly}. 
\begin{prop}\label{corkappapoly}
For every $m\ge0$, for fixed $n\ge1$, fixed $\dd'=(d_1,\dots,d_{n-1})\in(\ZZ_{\ge1})^n$, and $d_n\ge0$, we have
\begin{align}
&C(m;\dd) \,\sim\, \frac{1}{\pi}\sum_{k=0}^{\infty}\frac{C_k(m;\dd')}{X(m;\mathbf{d})^k} \qquad (X(m;\mathbf{d})\to\infty)\,,
\end{align}
where $\dd=(\dd',d_n)$, $C_k(m;\dd')$ are functions of $m,d_1,\dots,d_{n-1}$ with $C_0\equiv1$.
Moreover, there exist a sequence of polynomials 
		\begin{equation}
			c_k(m;p_1,p_2,\dots)\,\in\,\mathbb{Q}[p_1,p_2,\dots]\,, \quad k\geq 0\,,
	\end{equation} 
	with $c_0(m;p_1,p_2,\dots)\equiv1$, such that 
	\begin{equation}
		C_k(m;\dd') \= c_k(m;p_1(\dd'),p_2(\dd'),\dots)\,, \quad k\ge0\,.
	\end{equation}
Furthermore, under the degree assignments 
	\begin{equation}
		\deg \, p_i \= 2i+1\quad (i\geq 1)\,,
	\end{equation}
	the polynomials $c_k(m;p_1,p_2,\dots)$, $k\ge1$, satisfy the degree estimates
\begin{equation}
		\deg \, c_k(m;p_1,p_2,\dots) \,\leq\, k-1\,;
\end{equation} 
	in particular, $c_k(m;p_1,p_2,\dots)$ does not depend on any $p_{d}$ with $d\ge (k-1)/2$.
\end{prop}
\begin{proof}
Expanding both sides of~\eqref{Cmd} with respect to~$1/X(m;\dd)$, we obtain
\begin{align}\label{Ckmd}
&C_k(m;d_1,\dots,d_{n-1}) \=3^m\,\sum_{l=1}^{m}\frac{(-1)^{m-l}}{l!}\sum_{m_1,\dots,m_l\ge1\atop m_1+\cdots+m_l=m}
\binom{m}{m_1,\dots,m_l} \, \frac{1}{\prod_{i=1}^l (2m_i+1)!!} \nn\\
&\qquad\qquad\qquad\qquad\qquad\quad\times \sum_{j=0}^{k} C_j(d_1,\dots,d_{n-1},m_1,\dots,m_l) \, S(k-j,m-l-1)\,,\nn
\end{align}
where $S(n,k)$ are the Stirling numbers of the second kind. 
The proposition is then proved by using Theorem~\ref{thmpoly}.
\end{proof}

\end{document}